\documentclass[11pt, leqno, letterpaper]{amsart}
 
\usepackage{graphicx}    

\usepackage[margin=1.2in]{geometry}

\usepackage{amsmath,amsthm,amsfonts,amssymb}

\newtheorem{theorem}{Theorem}

\newtheorem{corollary}[theorem]{Corollary}

\newtheorem{lemma}[theorem]{Lemma}

\newtheorem{proposition}[theorem]{Proposition}

\newtheorem{remark}[theorem]{Remark}

\begin{document}

\title[Asymptotics for the Average of GBM and Asian Options]
{Asymptotics for the Discrete-Time Average of the Geometric Brownian Motion and Asian Options}

\author{Dan Pirjol}
\email{
dpirjol@gmail.com}

\author{Lingjiong Zhu}
\address
{Department of Mathematics \newline
\indent Florida State University \newline
\indent 1017 Academic Way \newline
\indent Tallahassee, FL-32306 \newline
\indent United States of America}
\email{
zhu@math.fsu.edu}

\date{15 September 2016}.


\subjclass[2010]{91G20,91G80, 60F05,60F10}
\keywords{Asian options, central limit theorems, Berry-Esseen bound, large deviations.}

\begin{abstract}
The time average of geometric Brownian motion plays a crucial role in the pricing of
Asian options in mathematical finance. In this paper we consider the asymptotics of
the discrete-time average of a geometric Brownian motion sampled on uniformly spaced
times in the limit of a very large number of averaging time steps. We derive almost sure
limit, fluctuations, large deviations, and also the asymptotics of the moment generating
function of the average. Based on these results, we derive the asymptotics for the price of
Asian options with discrete-time averaging in the Black-Scholes model, with both fixed
and floating strike.
\end{abstract}

\maketitle

\section{Introduction}

Asian (or average) options are widely traded instruments in the financial 
markets, which involve the time average of  the price of an asset $S_t$. 
Most commonly $S_t$ is a stock price or a commodity futures contract price, 
for example oil or natural gas futures. An Asian call option has payoff of 
the form
\begin{equation}
\mbox{Payoff} = \mbox{max}
\left\{\frac{1}{n} \sum_{i=1}^n S_{t_i} - K, 0\right\} \,,
\end{equation}
where $0 \leq t_1 < t_2 < \cdots < t_n$ is a sequence of strictly increasing
times, called sampling or averaging dates. Under risk-free neutral pricing, 
the price of such an option is given by the expectation of the payoff in the
risk-neutral measure. Assuming the Black-Scholes model one is led to
study the distributional properties of 
the discrete time average of the asset price
\begin{equation}\label{Adef}
A_n = \frac{1}{n} \sum_{i=1}^n S_{t_i}
\end{equation}
under the assumption that $S_n$ follows a geometric Brownian motion
\begin{equation}
dS_{t} = (r - q)S_{t}dt + \sigma S_{t}dZ_{t},
\end{equation}
where $Z_{t}$ is a standard Brownian motion, $r$ is the risk-free rate, $q$ is the dividend yield
and $\sigma$ is the volatility.

The main technical difficulty for pricing Asian options is that the probability
distribution of the discrete time average (\ref{Adef}) does not have a simple
expression. If the averaging times are uniformly distributed, the time average
can be well approximated, for sufficiently small time step, 
by a continuous average
\begin{equation}
A_n = \frac{1}{t_n} \int_0^{t_n} S_t dt\,.
\end{equation}
When $S_t$ follows a geometric Brownian motion, the problem is reduced to the
study of the distributional properties of the time integral of the geometric
Brownian motion, which has been extensively studied in the literature. 
See \cite{DufresneReview}
for a review of the main results and their applications to the Asian 
options pricing.

A wide variety of methods have been proposed for pricing Asian options,
and a brief survey is given below. 

1. PDE methods \cite{RogersShi,Vecer,VecerRisk,Lord}. The pricing of an Asian option
can be reduced to the solution of a 1+1 partial differential equation,
which is solved numerically. This method can be applied both to continuous-time 
and discrete-time averaging Asian options \cite{Andreasen}.
See also Alziary et al. \cite{Alziary}.

2. The Laplace transform method \cite{GemanYor,CS}: 
the Asian option price with random exponentially distributed maturity can be
found in closed form for the case when the asset price
$S_t$ follows a geometric Brownian motion. This reduces the problem of the 
Asian option pricing to the inversion of a Laplace transform. 

3. Spectral method \cite{Linetsky}. 
The probability distribution of the time integral of the geometric Brownian 
motion can be related to that of a Bessel process \cite{1,2}. 
The transition density of this Bessel process can be expanded in an 
eigenfunction series \cite{Wong}, and Asian option prices
can be evaluated using the eigenfunction expansion, truncated
to a sufficient high order \cite{Linetsky}.

4. Bounds and control variates methods. There is a large literature on
deriving bounds on Asian option prices. Both lower and upper bounds have been 
given, see \cite{Lord} for an overview. They can be used also in conjunction 
with Monte Carlo methods as control variates. One precise method of this type 
which is popular in practice was given by Curran \cite{Curran}. 
Other methods which take into account the discrete time averaging
have been proposed in \cite{disc1,disc2,disc3}.

5. Monte Carlo simulation. See e.g. Kemna and Vorst \cite{Kemna}, Fu, Madan,
Wang \cite{FMW},  Lapeyre and Teman \cite{Lapeyre}.

6. Analytical approximations. Various numerical methods have been proposed 
which approximate the distribution of the arithmetic average $A_n$ using 
parametric forms, such as log-normal \cite{Levy} or inverse Gamma 
distributions \cite{MP}. 

We note also the more general approach of \cite{Vecer2014} which can be 
applied for a wide class of models. 

Most of the theoretical results in the literature concerning the distribution 
of the time average of the geometric Brownian motion refer to the continuous 
time average.
The discrete sum of the geometric Brownian motion is a particular case of the 
sum of correlated log-normals which has been studied extensively in the 
literature, see \cite{ReviewLit} for an overview.
Dufresne has obtained in \cite{DufresneLN} limit distribution for the discrete
time average in the limit of very small volatility $\sigma \to 0$. 
A recent work by the present authors \cite{IME} studied the properties of the 
discrete time average at fixed $\sigma$ in the limit $n\to \infty$,
and its convergence to the continuous time average as the time step $\tau \to 0$.

In this paper, we concentrate on the discrete time average of the geometric 
Brownian motion, $A_{n}=\frac{1}{n}\sum_{i=1}^{n}S_{t_{i}}$.
We assume the Black-Scholes model, that is, the asset price follows a 
geometric Brownian motion
\begin{equation}
S_t = S_0 e^{\sigma Z_t + (r - q- \frac12\sigma^2) t} \,,
\end{equation}
where $Z_t$ is a standard Brownian motion. 
We would like to study the distributional properties of the average of the 
discretely sampled asset price (\ref{Adef}) defined on the discrete times 
uniformly spaced $t_i =i\tau$ with time step $\tau$.

We will derive in this paper asymptotic results about $A_n$ in the limit 
$n\to \infty$ by keeping fixed the following
combinations of model parameters
\begin{align}
&\beta = \frac12 \sigma^2 t_n n = \frac12 \sigma^2 \tau n^2\label{AssumpI} \,,
\\
&(r - q) \tau n = \rho\,.\label{AssumpII}
\end{align}
Note that $\beta$ is always positive but $\rho$ can be both positive and negative.
We also note that the conditions \eqref{AssumpI} and \eqref{AssumpII} can be
replaced by $\lim_{n\rightarrow\infty}\frac{1}{2}\sigma^{2}\tau n^{2}=\beta$
and $\lim_{n\rightarrow\infty}(r-q)\tau n=\rho$ and all the results
in this paper will still hold. 

The constraints \eqref{AssumpI}, \eqref{AssumpII} include two interesting 
regimes:
\begin{itemize}
\item
When the maturity $t_{n}=\tau n$ is constant, and so are the interest rates $r$ 
and dividend yield $q$, then, \eqref{AssumpI} assumes that the volatility 
$\sigma$ is of the order $O(\frac{1}{\sqrt{n}})$. 
Therefore, the conditions \eqref{AssumpI} and \eqref{AssumpII} include the 
\emph{small volatility} regime.
\item
When the maturity $t_{n}=\tau n$ is small, that is, $t_{n}\rightarrow 0$ as 
$n\rightarrow\infty$ and in particular is of the order $\frac{1}{n}$, then by 
\eqref{AssumpII}, the volatility $\sigma$ is a constant.
If the interest rates $r$ and dividend yield $q$ are constant, then 
\eqref{AssumpII} is replaced by
$\lim_{n\rightarrow\infty}(r-q)\tau n=0$, i.e., $\rho=0$. 
Therefore, the conditions \eqref{AssumpI} and \eqref{AssumpII} include the 
\emph{short maturity} regime.
\end{itemize}
We emphasize that we do not make any assumptions about the values of
$\rho,\beta$, and they can be arbitrary. The validity of our asymptotic 
results require only that $n\gg 1$, such that these regimes cover
most cases of practical interest, provided that the number of averaging times
is sufficiently large. 

We present in this paper three asymptotic results for the distributional 
properties of the discrete time average of a geometric Brownian motion in the 
limit of a large number of averaging time steps $n$: i) almost sure limit and 
fluctuation results for $A_n$, ii) an asymptotic result for the moment generating 
function of the partial sums $nA_n$ for $n\to \infty$, and iii) large deviations
results for $\mathbb{P}(A_{n}\in\cdot)$. Using these asymptotic results, 
we derive rigorously asymptotics for the prices of out-of-the-money, 
in-the-money and at-the-money Asian options. 

Section~\ref{sec:2} presents the almost sure and fluctuations results for $A_n$
in the $n\to \infty$ limit. Section~\ref{sec:3} presents an asymptotic result 
for the Laplace transform of the finite sum of the geometric
Brownian motion sampled on $n$ discrete times $nA_n$, in the limit $n \to \infty$.
In Section~\ref{sec:4} we consider the asymptotics of fixed strike Asian 
options following from the large deviations result iii), and in 
Section~\ref{sec:5} we treat the case of the floating strike Asian options.
These asymptotic results can be used to obtain approximative pricing formulas
for Asian options, and in Section~\ref{sec:6} we compare the numerical 
performance of the asymptotic result against alternative methods for pricing 
Asian options under the BS model. 
Some of the proposed methods are known to be less efficient numerically 
in the small maturity and/or small volatility limit \cite{GemanYor,Linetsky}. 
The asymptotic results derived in this paper are of practical interest as they
complement these approaches in a region where their numerical performance is
not very good. We demonstrate good agreement of our asymptotic results with 
alternative pricing methods for Asian options with realistic values of the 
model parameters.

\section{Asymptotics for the discrete time average of geometric Brownian motion}
\label{sec:2}

We have the almost sure limit:
\begin{proposition}\label{prop1}
We have
\begin{equation}
\lim_{n\to \infty} A_n = A_\infty \equiv
S_0 \frac{1}{\rho} (e^{\rho}-1)\mbox{  a.s.}\,.
\end{equation}
\end{proposition}

\begin{proof}
Note that $\max_{1\leq i\leq n}\sigma|Z_{t_{i}}|
=\max_{1\leq i\leq n}\sqrt{\frac{2\beta}{\tau}}\frac{1}{n}|Z_{i\tau}|$
and from the property of Brownian motion, $\frac{1}{n}\max_{1\leq i\leq n}|Z_{i\tau}|\rightarrow 0$
a.s. as $n\rightarrow\infty$.
Moreover, $\frac{1}{2}\sigma^{2}t_{i}=\frac{\beta i}{n^{2}}\leq\frac{\beta}{n}\rightarrow 0$ as $n\rightarrow\infty$
uniformly in $1\leq i\leq n$.
Therefore, $S_{t_{i}}$ can be approximated by $S_0 e^{(r-q)t_{i}}$ uniformly in $1\leq i\leq n$, 
that is, $\max_{1\leq i\leq n}|S_{t_{i}}-S_0 e^{(r-q)t_{i}}|\rightarrow 0$ a.s. as $n\rightarrow\infty$.
Finally, notice that
\begin{equation}
\frac{1}{n} \sum_{i=1}^{n}e^{(r-q)t_{i}} = 
\frac{1}{n}\sum_{i=1}^{n}e^{\rho \frac{i}{n}} =
\frac{1}{n}\frac{e^{\rho} -1}{1 - e^{-\frac{\rho}{n}}} \to 
\frac{1}{\rho}
(e^{\rho}-1)\,,\quad n \to \infty\,.
\end{equation}
Hence, we proved the desired result.
\end{proof}

We have also the following fluctuation result:
\begin{proposition}\label{prop2}
The time average $A_n$ converges in distribution to a normal
distribution in the $n\to \infty$ limit
\begin{equation}
\lim_{n\to \infty} \sqrt{n}\frac{A_n- A_\infty}{S_0} = 
N\left(0,2\beta v(\rho)\right)\,.
\end{equation}
with
\begin{equation}\label{vdef}
v(a) :=\frac{1}{a^{3}}\left[ae^{2a}-\frac{3}{2}e^{2a}+2e^{a}-\frac{1}{2}\right].
\end{equation}
\end{proposition}


\begin{proof}
We have
\begin{align}
&\sqrt{n}\frac{A_n- A_\infty}{S_0} 
\\
&=\frac{1}{\sqrt{n}}
\sum_{i=1}^n (e^{\sigma Z_i  + (r-q- \frac12 \sigma^2) t_i} 
- e^{\rho \frac{i}{n}})
+\left[\frac{1}{\sqrt{n}}\sum_{i=1}^{n}e^{\rho \frac{i}{n}}-\sqrt{n}\frac{e^{\rho}-1}{\rho}\right]
\\
&=\frac{1}{\sqrt{n}} 
\sum_{i=1}^n e^{\rho \frac{i}{n}}
(e^{\frac{\sqrt{2\beta}}{n} B_i 
- \beta\frac{i}{n^2}} - 1) 
+\left[\frac{1}{\sqrt{n}}\sum_{i=1}^{n}e^{\rho \frac{i}{n}}-\sqrt{n}\frac{e^{\rho}-1}{\rho}\right],
\label{TwoTerms}
\end{align}
where $Z_i = \sqrt{\tau} B_i$ with $B_i$ a standard Brownian motion.
We can rewrite the second term in \eqref{TwoTerms} as
\begin{align}
\frac{1}{\sqrt{n}}\sum_{i=1}^{n}e^{\rho \frac{i}{n}}-\sqrt{n}\frac{e^{\rho}-1}{\rho}
&=\frac{1}{\sqrt{n}}\frac{e^{\rho}-1}{1-e^{-\frac{\rho}{n}}}-\sqrt{n}\frac{e^{\rho}-1}{\rho}
\\
&=(e^{\rho}-1)\frac{1}{\sqrt{n}}
\left[\frac{1}{\frac{\rho}{n}-\frac{1}{2}\frac{\rho^{2}}{n^{2}}+O(n^{-3})}-
\frac{n}{\rho}\right]
\nonumber
\\
&=(e^{\rho}-1)\frac{1}{\sqrt{n}}\frac{\frac{1}{2}\frac{\rho^{2}}{n}+O(n^{-2})}
{\rho(\frac{\rho}{n}-\frac{1}{2}\frac{\rho^{2}}{n^{2}}+O(n^{-3}))}
\rightarrow 0\,,
\nonumber
\end{align}
as $n\rightarrow\infty$.

The first term in \eqref{TwoTerms} can be written further as
\begin{eqnarray}\label{8}
&& \frac{1}{\sqrt{n}} 
\sum_{i=1}^n e^{\rho \frac{i}{n}}
(e^{\frac{\sqrt{2\beta}}{n} B_i - \beta\frac{i}{n^2}} - 1) =
\frac{1}{\sqrt{n}} 
\sum_{i=1}^n e^{\rho \frac{i}{n}}
\frac{\sqrt{2\beta}}{n} B_i + \xi_n \,,
\end{eqnarray}
where we defined
\begin{equation}
\xi_n \equiv \frac{1}{\sqrt{n}}
\sum_{i=1}^n e^{\rho \frac{i}{n}}
\Big(e^{\frac{\sqrt{2\beta}}{n} B_i 
- \beta\frac{i}{n^2}} 
- \frac{\sqrt{2\beta}}{n} B_i- 1\Big) \,.
\end{equation}
We claim that $\xi_{n}\rightarrow 0$ in probability as $n\rightarrow\infty$.

We have the following upper bound on $\xi_n$.
\begin{equation}
\xi_n \leq \frac{1}{\sqrt{n}}
\sum_{i=1}^n e^{\rho \frac{i}{n}}
\left(e^{\frac{\sqrt{2\beta}}{n} B_i} 
- \frac{\sqrt{2\beta}}{n} B_i- 1\right) \equiv \xi_n^{\rm (up)} \,.
\end{equation}
The upper bound $\xi_n^{\rm (up)}$ is a non-negative random variable 
since $e^x-1-x \geq 0$ for any real $x$.
The expectation of $\xi_n^{\rm (up)}$ can be computed exactly
\begin{align}
\mathbb{E}[\xi_n^{\rm (up)}]
&=\frac{1}{\sqrt{n}}
\sum_{i=1}^n e^{\rho \frac{i}{n}}
(e^{\frac{\beta}{n^2}i} - 1) 
\\
&=\frac{1}{\sqrt{n}} 
\left( \frac{e^{\rho + \frac{\beta}{n}}-1}
{1 - e^{-\frac{\rho}{n} - \frac{\beta}{n^2}}} -
\frac{e^{\rho}-1}
{1 - e^{-\frac{\rho}{n}}}\right) 
=\frac{1}{\sqrt{n}}
\left( \frac{\beta}{\rho} + o(1/n) \right)\,.
\nonumber
\end{align}
This goes to zero as $n\to \infty$. The Markov inequality 
implies that $\xi_n^{\rm (up)} \to 0$ in probability as $n\rightarrow\infty$. 

Next, let us estimate the lower bound on $\xi_n$. We have
\begin{align}
\xi_n 
&\geq\frac{1}{\sqrt{n}} 
\sum_{i=1}^n e^{\rho \frac{i}{n}}
\left(e^{\frac{\sqrt{2\beta}}{n} B_i - \frac{\beta}{n}} 
-\left(\frac{\sqrt{2\beta}}{n}B_{i}-\frac{\beta}{n}\right)-1- \frac{\beta}{n}\right) 
\\
&\geq \frac{1}{\sqrt{n}} 
\sum_{i=1}^n e^{\rho \frac{i}{n}}\left(- \frac{\beta}{n}\right) 
= - \frac{\beta}{\sqrt{n}} \frac{e^{\rho}-1}{n(1-e^{-\frac{\rho}{n}})} 
\to 0\,, \nonumber
\end{align}
where we used again in the second step the inequality 
$e^{x} \geq 1 + x$.

The first term in (\ref{8}) is a normal random variable 
and converges in distribution to a normal distribution
with mean zero and variance to be determined.
\begin{equation}
\frac{1}{\sqrt{n}} \sum_{i=1}^n e^{\rho \frac{i}{n}}
\frac{\sqrt{2\beta}}{n} B_i \to N\left(0,2\beta v(\rho)\right)\,.
\end{equation}

This can be computed by writing $B_i = \sum_{j=0}^{i-1} V_j$
with $V_j \sim N(0,1)$ i.i.d. normally distributed random variables with mean
zero and unit variance. The sum can be written as
\begin{align}
\frac{1}{\sqrt{n}} \sum_{i=1}^n e^{\rho \frac{i}{n}}
\frac{\sqrt{2\beta}}{n} B_i 
&=\frac{\sqrt{2\beta}}{n^{3/2}}
\sum_{j=0}^{n-1} V_j \sum_{i=j+1}^{n} e^{\rho \frac{i}{n}} 
\\
&=\frac{\sqrt{2\beta}}{n^{3/2}}
\sum_{j=0}^{n-1} V_j \frac{1}{e^{\rho\frac{1}{n}}-1} 
\left\{ e^{\rho \frac{n+1}{n}}
- e^{\rho \frac{j+1}{n}} \right\}\,.\nonumber
\end{align}
We can compute the variance of this random variable as
\begin{align}
&\mbox{Var}\left(\frac{1}{\sqrt{n}} 
\sum_{i=1}^n e^{\rho \frac{i}{n}}\frac{\sqrt{2\beta}}{n} B_i \right)
=\frac{2\beta}{n^{3}}
\sum_{j=0}^{n-1} \frac{1}{(e^{\rho\frac{1}{n}}-1)^{2}} 
\left(e^{\rho \frac{n+1}{n}}
- e^{\rho \frac{j+1}{n}} \right)^{2} \\
&=2\beta\frac{1}{n^{2}(e^{\rho\frac{1}{n}}-1)^{2}} 
\sum_{j=0}^{n-1} 
\left(e^{\rho \frac{n+1}{n}}
- e^{\rho \frac{j+1}{n}} \right)^{2}\frac{1}{n}
\rightarrow
\frac{2\beta}{\rho^{2}}\int_{0}^{1}(e^{\rho}-e^{\rho x})^{2}dx,
\nonumber
\end{align}
as $n\rightarrow\infty$, where we can compute that
\begin{equation}
\frac{2\beta}{\rho^{2}}\int_{0}^{1}(e^{\rho}-e^{\rho x})^{2}dx
=\frac{2\beta}{\rho^{3}}\left[\rho e^{2\rho}-\frac{3}{2}e^{2\rho}+2e^{\rho}-\frac{1}{2}\right].
\end{equation}
\end{proof}


\section{Moment generating function}
\label{sec:3}

Define the moment generating function of $n A_n$ as
\begin{equation}
F_n(\theta) := \mathbb{E}[e^{\theta n A_n}] \,.
\end{equation}
For $\theta < 0$, this is the Laplace transform of the distribution 
function of $n A_n$. 

We are interested in the limit $\lim_{n\rightarrow\infty}\frac{1}{n}\log F_{n}(\theta)$.
We will compute this limit using the theory of large deviations.
Before we proceed, recall that a sequence $(P_{n})_{n\in\mathbb{N}}$ of probability measures on a topological space $X$ 
satisfies the large deviation principle with rate function $I:X\rightarrow\mathbb{R}$ if $I$ is non-negative, 
lower semicontinuous and for any measurable set $A$, we have
\begin{equation}
-\inf_{x\in A^{o}}I(x)\leq\liminf_{n\rightarrow\infty}\frac{1}{n}\log P_{n}(A)
\leq\limsup_{n\rightarrow\infty}\frac{1}{n}\log P_{n}(A)\leq-\inf_{x\in\overline{A}}I(x).
\end{equation}
Here, $A^{o}$ is the interior of $A$ and $\overline{A}$ is its closure. 
The rate function $I$ is said to be good if for any $m$, the level set $\{x:I(x)\leq m\}$ is compact.
We refer to Dembo and Zeitouni \cite{Dembo} or Varadhan \cite{VaradhanII} for general background of large deviations and the applications. 

We have the following limit theorem for the generating function in the limit 
$n\to \infty$ at fixed $\beta$.

\begin{theorem}\label{theorem3}
For any $\theta>0$, $F_{n}(\theta)=\infty$ and for any $\theta\leq 0$,
\begin{align}
\lim_{n\rightarrow\infty}\frac{1}{n} \log F_n(\theta) 
=\sup_{g\in\mathcal{AC}_{0}[0,1]}
\left\{\theta S_{0}\int_{0}^{1}e^{\sqrt{2\beta}g(x)}dx
-\frac{1}{2}\int_{0}^{1}\left(g'(x)-\frac{\rho}{\sqrt{2\beta}}\right)^{2}dx
\right\}\,.
\end{align}
\end{theorem}

\begin{proof}
Since $\mathbb{E}[e^{\theta X}]=\infty$ for any $\theta>0$ for any log-normal random variable $X$, 
it is clear that $\mathbb{E}[e^{\theta nA_{n}}]=\infty$ for any $\theta>0$.
Next, for any $\theta\leq 0$,
\begin{align}
\mathbb{E}[e^{\theta nA_{n}}]
&=\mathbb{E}\left[e^{\theta\sum_{k=0}^{n-1}S_{0}
e^{\sigma Z_{t_{k}} + (r-q-\frac{1}{2}\sigma^{2})t_{k}}}\right]
\\
&=\mathbb{E}\left[e^{\theta S_{0}\sum_{k=0}^{n-1}
e^{\sigma\sqrt{\tau}\sum_{j=1}^{k}V_{j} + (r-q-\frac{1}{2}\sigma^{2})k\tau}}\right]
\nonumber
\\
&=\mathbb{E}\left[e^{\theta S_{0}\sum_{k=0}^{n-1}
e^{\frac{\sqrt{2\beta}}{n}\sum_{j=1}^{k}V_{j} + \frac{\rho k}{n}-\frac{\beta}{n^{2}}k}}\right]
\nonumber
\\
&=\mathbb{E}\left[e^{\theta S_{0}\sum_{k=0}^{n-1}
e^{\frac{\sqrt{2\beta}}{n}\sum_{j=1}^{k}(V_{j} + \frac{\rho}{\sqrt{2\beta}})-\frac{\beta}{n^{2}}k}}\right],
\nonumber
\end{align}
where $V_{j}:=\frac{1}{\sqrt{\tau}}(Z_{j}-Z_{j-1})$, $1\leq j\leq k$, are i.i.d. $N(0,1)$ random variables.
Note that $\sum_{j=1}^{0}V_{j}$ is defined as $0$.
By Mogulskii theorem, see e.g. \cite{Dembo}, 
$\mathbb{P}(\frac{1}{n}\sum_{j=1}^{\lfloor\cdot n\rfloor}(V_{j}+\frac{\rho}{\sqrt{2\beta}})\in\cdot)$
satisfies a large deviation principle on $L^{\infty}[0,1]$ with the good rate function
\begin{equation}
I(g)=\int_{0}^{1}\Lambda(g'(x))dx,
\end{equation}
if $g\in\mathcal{AC}_{0}[0,1]$, i.e., absolutely continuous and $g(0)=0$ 
and $I(g)=+\infty$ otherwise, where
\begin{equation}
\Lambda(x):=\sup_{\theta\in\mathbb{R}}
\left\{\theta x-\log\mathbb{E}\left[e^{\theta(V_{1}+\frac{\rho}{\sqrt{2\beta}})}\right]\right\}
=\frac{1}{2}\left(x-\frac{\rho}{\sqrt{2\beta}}\right)^{2}.
\end{equation}
Let $g(x):=\frac{1}{n}\sum_{j=1}^{\lfloor xn\rfloor}(V_{j}+\frac{\rho}{\sqrt{2\beta}})$. 
Then, 
\begin{equation}
\int_{0}^{1}e^{\sqrt{2\beta}g(x)}dx
=\sum_{k=0}^{n-1}\int_{\frac{k}{n}}^{\frac{k+1}{n}}e^{\sqrt{2\beta}g(x)}dx
=\frac{1}{n}\sum_{k=0}^{n-1}e^{\frac{\sqrt{2\beta}}{n}\sum_{j=1}^{k}(V_{j}+\frac{\rho}{\sqrt{2\beta}})}.
\end{equation}
Moreover, we claim that
\begin{equation}
g\mapsto\int_{0}^{1}e^{\sqrt{2\beta}g(x)}dx
\end{equation}
is a continuous map. 
Let $g_{n}$ be any sequence in $L^{\infty}[0,1]$ so that
$g_{n}\rightarrow g$ in $L^{\infty}[0,1]$. 
Observe that for any $|x|\leq\frac{1}{2}$.
\begin{align}
|e^{x}-1| = \left|x+\frac{x^{2}}{2!}+\frac{x^{3}}{3!}+\cdots\right|
\leq |x|(1+|x|+|x|^{2}+\cdots) \leq 2|x|.
\end{align}
Let $n$ be sufficiently large so that $\sqrt{2\beta}\Vert g_{n}-g\Vert_{L^{\infty}[0,1]}\leq\frac{1}{2}$. 
Therefore, we have
\begin{align}
&\left|\int_{0}^{1}e^{\sqrt{2\beta}g_{n}(x)}dx
-\int_{0}^{1}e^{\sqrt{2\beta}g(x)}dx\right|
=\left|\int_{0}^{1}e^{\sqrt{2\beta}g(x)}\left(e^{\sqrt{2\beta}(g_{n}(x)-g(x))}
-1\right)dx\right|
\nonumber
\\
&\leq e^{\sqrt{2\beta}\Vert g\Vert_{L^{\infty}[0,1]}}
\int_{0}^{1}\left|e^{\sqrt{2\beta}(g_{n}(x)-g(x))}
-1\right|dx
\leq 2\sqrt{2\beta}e^{\sqrt{2\beta}\Vert g\Vert_{L^{\infty}[0,1]}}\Vert g_{n}-g\Vert_{L^{\infty}[0,1]}
\nonumber
\end{align}
which converges to $0$ as $n\rightarrow\infty$. Hence the map is continuous. 
Let us recall the celebrated Varadhan's lemma from large deviations theory, see e.g. \cite{Dembo}.
if $\mathbb{P}(Z_{n}\in\cdot)$ satisfies a large deviation
principle with good rate function $I:\mathcal{X}\rightarrow[0,+\infty]$, and if $\phi$ is a continuous map and
\begin{equation}\label{superexp}
\lim_{M\rightarrow+\infty}\limsup_{n\rightarrow\infty}
\frac{1}{n}\log\mathbb{E}\left[e^{n\phi(Z_{n})}1_{\phi(Z_{n})\geq M}\right]=-\infty,
\end{equation}
then
\begin{equation*}
\lim_{n\rightarrow\infty}\frac{1}{n}\log
\mathbb{E}[e^{n\phi(Z_{n})}]=\sup_{x\in\mathcal{X}}\{\phi(x)-I(x)\}.
\end{equation*}
In our case, 
\begin{equation*}
\phi(g)=\theta S_{0}\int_{0}^{1}e^{\sqrt{2\beta}g(x)}dx
\end{equation*}
is a continuous map. Moreover, for $\theta\leq 0$, $\phi(g)\leq 0$
and thus the condition \eqref{superexp} is trivially satisfied. 
Hence we can apply the Varadhan's lemma and get,
\begin{align}
&\lim_{n\rightarrow\infty}\frac{1}{n}\log\mathbb{E}\left[e^{\theta S_{0}\sum_{k=0}^{n-1}
e^{\frac{\sqrt{2\beta}}{n}\sum_{j=1}^{k}(V_{j}+\frac{\rho}{\sqrt{2\beta}})}}\right]
\\
&=\sup_{g\in\mathcal{AC}_{0}[0,1]}\left\{\theta S_{0}\int_{0}^{1}e^{\sqrt{2\beta}g(x)}dx
-\frac{1}{2}\int_{0}^{1}\left(g'(x)-\frac{\rho}{\sqrt{2\beta}}\right)^{2}dx\right\}.
\nonumber
\end{align}
Finally, notice that
\begin{align}
\mathbb{E}\left[e^{\theta S_{0}\sum_{k=0}^{n-1}
e^{\frac{\sqrt{2\beta}}{n}\sum_{j=1}^{k}(V_{j}+\frac{\rho}{\sqrt{2\beta}})}}\right]
&\leq \mathbb{E}[e^{\theta nA_{n}}]
\\
&\leq\mathbb{E}\left[e^{\theta S_{0}e^{-\frac{\beta}{n}}\sum_{k=0}^{n-1}
e^{\frac{\sqrt{2\beta}}{n}\sum_{j=1}^{k}(V_{j}+\frac{\rho}{\sqrt{2\beta}})}}\right].
\nonumber
\end{align}
Hence, for any $\theta\leq 0$,
\begin{align}
&\lim_{n\rightarrow\infty}\frac{1}{n}\log\mathbb{E}[e^{\theta nA_{n}}]
\\
&=\sup_{g\in\mathcal{AC}_{0}[0,1]}\left\{\theta S_{0}\int_{0}^{1}e^{\sqrt{2\beta}g(x)}dx
-\frac{1}{2}\int_{0}^{1}\left(g'(x)-\frac{\rho}{\sqrt{2\beta}}\right)^{2}dx\right\}.
\nonumber
\end{align}
\end{proof}

\subsection{Solution of the Variational Problem}
The variational problem in Theorem~\ref{theorem3} can be re-stated as
\begin{equation}
\lim_{n\to \infty} \frac{1}{n}\log F_n(\theta) = 
\lambda(-\theta S_0,\sqrt{2\beta};\rho)
\end{equation}
where $\lambda(a,b;\rho)$ is the solution of the variational problem
\begin{equation}\label{varprob40}
\lambda(a,b;\rho) =\sup_{g\in\mathcal{AC}_{0}[0,1]}
\left\{ - a\int_{0}^{1}e^{bg(x)}dx
-\frac{1}{2}\int_{0}^{1}\left(g'(x)-\frac{\rho}{b}\right)^{2}dx
\right\}\,.
\end{equation}
Here we have $a,b>0$.

This variational problem can be solved explicitly,
and the solution is given by the following result.


\begin{proposition}\label{VarProblem}
The function $\lambda(a,b;\rho)$ is given by one of the two expressions
\begin{align}\label{lam1}
\lambda_1(a,b;\rho)
&= a \left\{ 1 + \sinh^2\left(\frac{\delta}{2}\right) \left( 1 - \frac{4\rho}{\delta^2}
+ \frac{\rho^2}{\delta^2} \right) - \frac{2-\rho}{\delta} \sinh \delta \right\}\\
&
\qquad
+ \frac{2}{b^2}\rho \log 
\left[\cosh\left(\frac{\delta}{2}\right) + \frac{\rho}{\delta} \sinh\left(\frac{\delta}{2}\right) \right]
- \frac{\rho^2}{b^2}, \nonumber 
\end{align}
or 
\begin{align}\label{lam2}
\lambda_2(a,b;\rho) &=\
a \left\{ 1 - \sin^2 \xi \left(1 + \frac{\rho}{\xi^2} - 
\frac{\rho^2}{4\xi^2}\right) + \frac{\rho-2}{2\xi} \sin(2\xi)
\right\} \\
&
\qquad
+\frac{2\rho}{b^2} \log \left[\cos \xi + \frac{\rho}{2\xi}\sin\xi\right]
- \frac{\rho^2}{b^2} \,. \nonumber
\end{align}
In (\ref{lam1}) $\delta$ is the solution of the equation
\begin{equation}\label{deltaeq0}
\rho^2 - \delta^2 = 2ab^2 \left( \cosh\left(\frac12\delta\right) + \frac{\rho}{\delta}
\sinh\left(\frac12\delta\right) \right)^2\,,
\end{equation}
and in (\ref{lam2}) $\xi$ is the unique solution 
$\xi \in (0,\xi_{\rm max})$ of the equation
\begin{equation}\label{lambdaeq0}
2\xi^2 (4\xi^2 + \rho^2) = ab^2 
(2\xi\cos\xi + \rho \sin\xi)^2 \,.
\end{equation}
$\xi_{\rm max}$ is the smallest solution of the equation 
$\tan \xi_{\rm max} = - 2\xi_{\rm max}/\rho$.

For given $(a>0,b,\rho)$ only one of the two equations \eqref{deltaeq0} and 
\eqref{lambdaeq0} has a solution, such
that the solution of the variational problem is unique.
\end{proposition}


\begin{proof}
The proof will be given in the Appendix.
\end{proof}

Let us recall that $\lim_{n\rightarrow\infty}(r-q)\tau n=\rho$, and 
in the short maturity limit $t_n\to 0$ at constant $r,q$, we have $\rho=0$.
Therefore, the special case $\rho=0$ is of practical interest
when considering the short maturity limit.
For this case it is clear that only (\ref{lambdaeq0}) has a solution for 
$a>0$ so we get the simpler result.
\begin{corollary}
The function $\lambda(a,b;0)$ in the $\rho=0$ limit is given by
\begin{equation}
\lambda(a,b;0)=a \left( \cos^2 \xi  - \frac{1}{\xi} \sin(2\xi)
\right)\,,
\end{equation}
where $\xi$ is the solution of the equation
\begin{eqnarray}
2\xi^2 = ab^2 \cos^2\xi \,.
\end{eqnarray}
\end{corollary}

In conclusion, the result of Theorem~\ref{theorem3} and 
Proposition~\ref{VarProblem} gives an asymptotic expression for the Laplace 
transform of the discrete sum of
the geometric Brownian motion $nA_n$ in the limit $n\to \infty$, of the 
form $\mathbb{E}[e^{-\theta n A_n}] = 
\exp(n\lambda(\theta S_0,\sqrt{2\beta};\rho) + o(n))$.
This result could be used for numerical simulations of $nA_n$,
similar to the approach presented in \cite{Laplace} using an asymptotic
result for the Laplace transform of the sum of correlated log-normals.
Another possible application would be to obtain a first-order approximation 
of Asian options prices in the asymptotic $n \gg 1$ limit using the Carr-Madan 
formula \cite{CarrMadan}. 

In the next Section we present the leading asymptotics for the Asian option 
prices using the theory of large deviations. 

\section{Asymptotics for Asian options prices}
\label{sec:4}

Asymptotics for the option pricing is a well studied subject in mathematical finance.
There is a vast literature on the asymptotics for option pricing, especially
the asymptotics for the vanilla option pricing and the corresponding implied 
volatility for various continuous-time models,
see e.g. \cite{Berestycki,GHLPW,Feng,Forde,Tehranchi}. 
We are interested in the asymptotics for the pricing of the Asian options in 
the discrete time setting, under the assumptions \eqref{AssumpI} and 
\eqref{AssumpII}.

Let us consider an Asian option with strike price $K$, in the Black-Scholes
model with volatility $\sigma$, risk free rate $r$ and dividend yield $q$. 
The prices of the put and call options at time zero are given by
\begin{align}
&P(n):=e^{-rt_{n}}\mathbb{E}[(K-A_{n})^{+}],
\\
&C(n):=e^{-rt_{n}}\mathbb{E}[(A_{n}-K)^{+}],
\end{align}
respectively, where $A_{n}=\frac{1}{n}\sum_{i=1}^{n}S_{t_{i}}$ and 
the expectation are taken under the risk-neutral
probability measure under which the asset price satisfies the SDE 
$dS_{t}=(r-q)S_{t}dt+\sigma S_{t}dW_{t}$.
Also notice that $e^{-rt_{n}}=e^{-\frac{r}{r-q}(r-q)\tau n}=e^{-\frac{r}{r-q}\rho}$.
Recall that we have proved that
$A_{n}\rightarrow A_{\infty}=\frac{S_{0}}{\rho}(e^{\rho}-1)$ a.s. as $n\rightarrow\infty$.
Since $(K-A_{n})^{+}\leq K$, by the bounded convergence theorem from real 
analysis, we have
\begin{equation}
\lim_{n\rightarrow\infty}P(n)=e^{-\frac{r}{r-q}\rho}\lim_{n\rightarrow\infty}\mathbb{E}[(K-A_{n})^{+}]
=e^{-\frac{r}{r-q}\rho}\left(K-\frac{S_{0}}{\rho}(e^{\rho}-1)\right)^{+}.
\end{equation}
From put-call parity,
\begin{align}
C(n)-P(n)&=e^{-rt_{n}}\mathbb{E}[A_{n}-K]
=e^{-\frac{r}{r-q}\rho}\left[\frac{1}{n}\sum_{i=1}^{n}\mathbb{E}[S_{t_{i}}]-K\right]
\nonumber
\\
&=e^{-\frac{r}{r-q}\rho}\left[\frac{1}{n}S_0\sum_{i=1}^{n}e^{\rho\frac{i}{n}}-K\right]
\rightarrow e^{-\frac{r}{r-q}\rho}\left(\frac{S_{0}}{\rho}(e^{\rho}-1)-K\right),
\nonumber
\end{align}
as $n\rightarrow\infty$. Therefore,
\begin{equation}
\lim_{n\rightarrow\infty}C(n)
=e^{-\frac{r}{r-q}\rho}\left(\frac{S_{0}}{\rho}(e^{\rho}-1)-K\right)^{+}.
\end{equation}

\subsection{Out-of-the-Money Case}

When $K<\frac{S_{0}}{\rho}(e^{\rho}-1)$, $\lim_{n\rightarrow\infty}P(n)=0$
and the put option is out-of-the-money and the decaying rate of $P(n)$ to zero
as $n\rightarrow\infty$ is governed by the left tail of the large deviations
of $A_{n}$. When $K>\frac{S_{0}}{\rho}(e^{\rho}-1)$,  
$\lim_{n\rightarrow\infty}C(n)=0$
and the call option is out-of-the-money and the decaying rate of $C(n)$ to zero
as $n\rightarrow\infty$ is governed by the right tail of the large deviations of $A_{n}$.
Before we proceed, let us first derive the large deviation principle for $\mathbb{P}(A_{n}\in\cdot)$.

\begin{proposition}\label{LDPProp}
$\mathbb{P}(A_{n}\in\cdot)$ satisfies a large deviation principle with rate function
\begin{equation}\label{RateFunction}
\mathcal{I}(x)=\inf_{g\in\mathcal{AC}_{0}[0,1],\int_{0}^{1}e^{\sqrt{2\beta}g(y)}dy=\frac{x}{S_{0}}}
\frac{1}{2}\int_{0}^{1}\left(g'(x)-\frac{\rho}{\sqrt{2\beta}}\right)^{2}dx,
\end{equation}
for $x\geq 0$ and $\mathcal{I}(x)=+\infty$ otherwise.
\end{proposition}

\begin{proof}
We proved already that $\frac{1}{n}\sum_{k=0}^{n-1}
e^{\frac{\sqrt{2\beta}}{n}\sum_{j=1}^{k}(V_{j}+\frac{\rho}{\sqrt{2\beta}})}
=\int_{0}^{1}e^{\sqrt{2\beta}g(x)}dx$, where $g(x)=\frac{1}{n}\sum_{j=1}^{\lfloor xn\rfloor}
(V_{j}+\frac{\rho}{\sqrt{2\beta}})$ and the map $g\mapsto\int_{0}^{1}e^{\sqrt{2\beta}g(x)}dx$
is continuous in the supremum norm. 
Since $\mathbb{P}(\frac{1}{n}
\sum_{j=1}^{\lfloor\cdot n\rfloor}(V_{j}+\frac{\rho}{\sqrt{2\beta}})\in\cdot)$ 
satisfies a large deviation principle
on $L^{\infty}[0,1]$ with rate function $\frac{1}{2}\int_{0}^{1}\left(g'(x)-\frac{\rho}{\sqrt{2\beta}}\right)^{2}dx$
if $g\in\mathcal{AC}_{0}[0,1]$ and $+\infty$ otherwise. From the contraction principle,
and the fact that $e^{-\frac{\beta}{n}}\leq e^{-\frac{\beta}{n^{2}}k}\leq 1$ uniformly in $0\leq k\leq n-1$,
we conclude that $\mathbb{P}(A_{n}\in\cdot)$ satisfies a large deviation principle
with rate function defined in \eqref{RateFunction}. Finally, notice that $A_{n}$ is positive and thus
$\mathcal{I}(x)=+\infty$ for any $x<0$.
\end{proof}

\begin{remark}
$\mathcal{I}(x)=0$ in \eqref{RateFunction} if and only if the optimal $g$ satisfies $g'(x)=\frac{\rho}{\sqrt{2\beta}}$
which is equivalent to $g(x)=\frac{\rho}{\sqrt{2\beta}}x$ since $g(0)=0$. This gives us
$\int_{0}^{1}e^{\sqrt{2\beta}g(y)}dy=\int_{0}^{1}e^{\rho y}dy=\frac{e^{\rho}-1}{\rho}$.
Thus $\mathcal{I}(x)=0$ if and only if $x=S_{0}\frac{e^{\rho}-1}{\rho}=A_{\infty}$ which
is consistent with the a.s. limit of $A_{n}$ as $n\rightarrow\infty$.
\end{remark}

\begin{remark}
We have proved that $\Gamma(\theta):=\lim_{n\rightarrow\infty}\frac{1}{n}\log\mathbb{E}[e^{\theta nA_{n}}]$ exists
for any $\theta\leq 0$ and is differentiable and $\Gamma(\theta)=+\infty$ for any $\theta>0$.
Since $\Gamma(\theta)=+\infty$ for any $\theta>0$, we cannot use G\"{a}rtner-Ellis theorem
to obtain large deviations for $\mathbb{P}(A_{n}\in\cdot)$. One may speculate that we might have subexponential
tails. But the intriguing fact is that we still have large deviations as stated in Proposition \ref{LDPProp}.
\end{remark}

We can further analyze and solve the variational problem 
\eqref{RateFunction}. For $\rho \neq 0$, the solution is given by the following
result.

\begin{proposition}\label{prop:rhogeneral}
The rate function of the discrete time average of the geometric Brownian 
motion is given by
\begin{equation}\label{Jcaldef}
\mathcal{I}(x) = \frac{1}{2\beta} \mathcal{J}(x/S_0,\rho),
\end{equation}
with
\begin{equation}
\mathcal{J}(x/S_0,\rho) =
\begin{cases}
\mathcal{J}_1(x/S_0,\rho) & \, \frac{x}{S_0} \geq 1 + \frac12\rho \\
\mathcal{J}_2(x/S_0,\rho) & \, \frac{x}{S_0} \leq 1 + \frac12\rho
\end{cases},
\end{equation}
where
\begin{align}
& \mathcal{J}_1\left(\frac{x}{S_0},\rho\right) =
\frac12 (\delta^2 - \rho^2) \left(1 - \frac{2\tanh(\frac12\delta)}
{\delta + \rho \tanh(\frac12\delta)} \right) \\
& \qquad\qquad\qquad\qquad\qquad - 
2\rho \log \left[\cosh\left(\frac12\delta\right) + 
\frac{\rho}{\delta} \sinh\left(\frac12 \delta\right) \right] + \rho^2 ,\nonumber \\
& \mathcal{J}_2\left(\frac{x}{S_0},\rho\right) =
2\left(\xi^2 + \frac14 \rho^2\right) 
\left\{
\frac{\tan\xi}{\xi + \frac{\rho}{2}\tan\xi}  - 1
 \right\} - 
2\rho \log\left( \cos\xi + \frac{\rho}{2\xi} \sin\xi \right) 
+ \rho^2 ,
\end{align}
and $\delta,\xi$ are the solutions of the equations
\begin{equation}\label{deltaeq2}
\frac{1}{\delta} \sinh \delta + 
\frac{2\rho}{\delta^2} \sinh^2\left(\frac12 \delta\right) = \frac{x}{S_0}\,.
\end{equation}
and
\begin{equation}\label{lambdaeq2}
\frac{1}{2\xi} \sin(2\xi) 
\left( 1 + \frac{\rho}{2} \frac{\tan\xi}{\xi} \right)
= \frac{x}{S_0}\,.
\end{equation}
\end{proposition} 

\begin{proof}
The proof is given in the Appendix.
\end{proof}

\begin{remark}
We note that the equations for $\mathcal{J}_{1,2}(K/S_0,\rho)$ can be put into
a unique form by denoting $z=2\xi = i\delta$. Expressed in terms of this
variable we have
\begin{align}
& \mathcal{J}\left(\frac{x}{S_0},\rho\right) =
\frac12 (z^2 + \rho^2) \left(\frac{2\tan\left(\frac{z}{2}\right)}
{z + \rho \tan\left(\frac{z}{2}\right)} - 1 \right) \\
& \qquad\qquad\qquad\qquad\qquad - 
2\rho \log \left[\cos\left(\frac{z}{2}\right) + 
\frac{\rho}{z} \sin\left(\frac{z}{2}\right) \right] + \rho^2 \,,\nonumber
\end{align}
where $z$ is the solution of the equation
\begin{eqnarray}
\frac{1}{z} \sin z + \frac{2\rho}{z^2} \sin^2\left(\frac{z}{2}\right) = 
\frac{x}{S_0}\,.
\end{eqnarray}
\end{remark}

\begin{remark}
The rate function $\mathcal{J}(K/S_0,\rho)$ vanishes for 
$K = A_\infty = S_0\frac{1}{\rho}(e^\rho - 1)$, as expected from the general 
properties of the rate function. Since we have 
$\frac{1}{\rho}(e^\rho-1) \geq 1 + \frac12\rho$
for any $\rho \in \mathbb{R}$, this zero occurs for $\mathcal{J}_1(K/S_0,\rho)$.
We note that the rate function 
$\mathcal{J}_1(K/S_0,\rho)$ vanishes at $\delta = \pm \rho$. 
Both these values of $\delta$ satisfy (\ref{deltaeq2}) with 
$K/S_0 = \frac{1}{\rho}(e^\rho-1)$, which corresponds to $K=A_\infty$. 
However, the true solution of the variational problem \eqref{RateFunction}
corresponds to $\delta = - \rho$ which gives the optimal function $g(x)=
\frac{\rho x}{\sqrt{2\beta}}$, see (\ref{f1}).
\end{remark}


For $\rho=0$ the solution to the variational problem \eqref{RateFunction}
simplifies and is given as follows.

\begin{corollary}\label{VarRateFunction}
For the special case $\rho=0$,
\begin{equation}\label{Icalresult}
\mathcal{J}(x)=
\begin{cases}
\frac12\delta^2 - \delta \tanh\left(\frac12 \delta \right)  & \frac{x}{S_0} \geq 1,
\\
2\xi (\tan\xi - \xi) &  0 < \frac{x}{S_0} \leq 1.
\end{cases}
\end{equation}
and $\mathcal{J}(x)=\infty$ otherwise,
where $\delta$ is the unique solution of the equation
\begin{equation}
\frac{1}{\delta} \sinh \delta = \frac{x}{S_0} ,
\end{equation}
and $\xi$ is the unique solution in $(0,\frac12\pi)$ of the equation
\begin{eqnarray}
\frac{1}{2\xi} \sin (2\xi) = \frac{x}{S_0} \,.
\end{eqnarray}
\end{corollary}

It can be shown that this is identical to the rate function for the short 
maturity asymptotics of Asian options with continuous time averaging in the 
Black-Scholes model \cite{SMAO}. 
The rate function $\mathcal{J}(K/S_0,\rho)$ can be evaluated numerically
using the result of Proposition~\ref{prop:rhogeneral}. 
The plot of $\mathcal{J}(x/S_0,\rho)$ is shown in Figure~\ref{Fig:Ical}
for $\rho = 0, 0.1$.

\begin{figure}[t]
\centering
\includegraphics[width=4.5in]{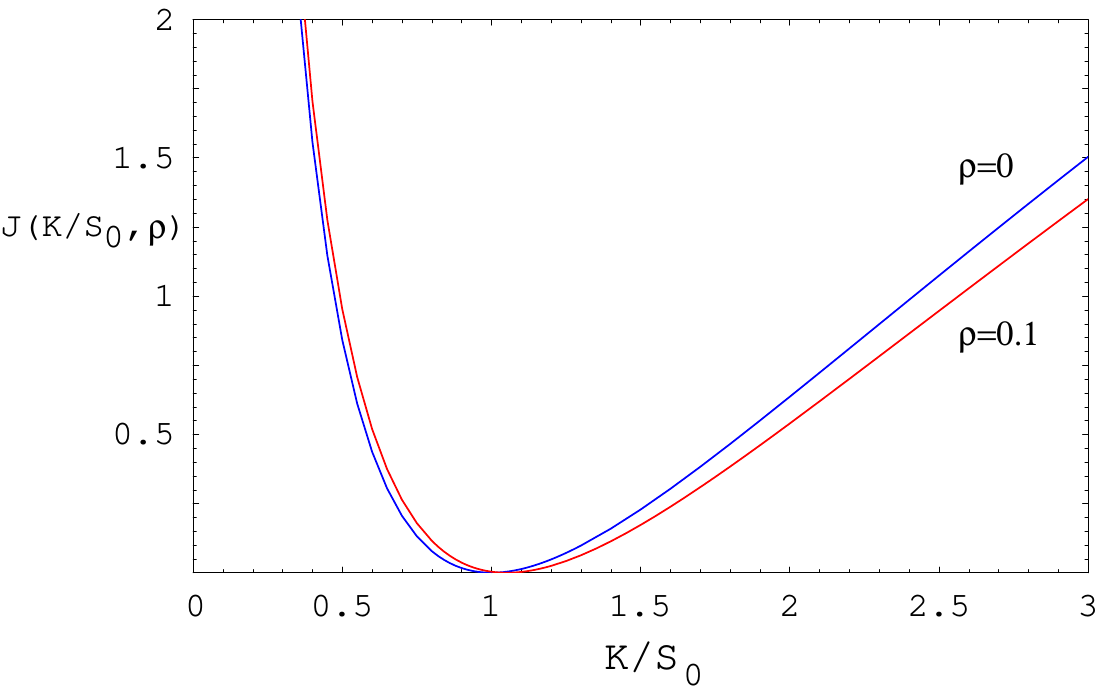}
\caption{Plot of the rate function $\mathcal{J}(K/S_0,\rho)$ vs $K/S_0$
for two values of the $\rho$ parameter $\rho=0, 0.1$. 
This is related to the rate function $\mathcal{I}(x)$ 
for the large deviations of the average of the geometric Brownian motion $A_n$
as in (\ref{Jcaldef}).}
\label{Fig:Ical}
\end{figure}

Using the large deviations results for $\mathbb{P}(A_{n}\in\cdot)$, we can 
obtain the asymptotics of the out-of-the-money Asian options prices. 
This is given by the following result.

\begin{proposition}\label{prop:13}
When $K<\frac{S_{0}}{\rho}(e^{\rho}-1)$, 
\begin{equation}\label{PLimit}
P(n)=e^{-n\mathcal{I}(K)+o(n)},\qquad\text{as $n\rightarrow\infty$},
\end{equation}
and when $K>\frac{S_{0}}{\rho}(e^{\rho}-1)$, 
\begin{equation}\label{CLimit}
C(n)=e^{-n\mathcal{I}(K)+o(n)},\qquad\text{as $n\rightarrow\infty$},
\end{equation}
where $\mathcal{I}(\cdot)$ was defined in \eqref{RateFunction}.
\end{proposition}

\begin{proof}
For any $0<\epsilon<K$,
\begin{align}
P(n)&\geq e^{-\frac{r}{r-q}\rho}\mathbb{E}\left[(K-A_{n})1_{A_{n}\leq K-\epsilon}\right]
\geq e^{-\frac{r}{r-q}\rho}\epsilon\mathbb{P}(A_{n}\leq K-\epsilon)\,.
\end{align}
Therefore, $\liminf_{n\rightarrow\infty}\frac{1}{n}\log P(n)\geq-\mathcal{I}(K-\epsilon)$.
Since it holds for any $\epsilon\in(0,K)$, we conclude that
\begin{equation}
\liminf_{n\rightarrow\infty}\frac{1}{n}\log P(n)\geq-\mathcal{I}(K).
\end{equation}
On the other hand,
\begin{equation}
P(n)=e^{-\frac{r}{r-q}\rho}\mathbb{E}\left[(K-A_{n})1_{A_{n}\leq K}\right]
\leq e^{-\frac{r}{r-q}\rho}K\mathbb{P}(A_{n}\leq K),
\end{equation}
which implies that $\limsup_{n\rightarrow\infty}\frac{1}{n}\log P(n)\leq-\mathcal{I}(K)$.
Hence, we proved the \eqref{PLimit}.

For any $\epsilon>0$,
\begin{align}
C(n)&\geq e^{-\frac{r}{r-q}\rho}\mathbb{E}\left[(A_{n}-K)1_{A_{n}\geq K+\epsilon}\right]
\geq e^{-\frac{r}{r-q}\rho}\epsilon\mathbb{P}(A_{n}\geq K+\epsilon)\,.
\end{align}
Therefore, $\liminf_{n\rightarrow\infty}\frac{1}{n}\log C(n)\geq-\mathcal{I}(K+\epsilon)$.
Since it holds for any $\epsilon>0$, we have
\begin{equation}
\liminf_{n\rightarrow\infty}\frac{1}{n}\log C(n)\geq-\mathcal{I}(K).
\end{equation}
For any $\frac{1}{p}+\frac{1}{q}=1$, $p,q>1$, by H\"{o}lder's inequality,
\begin{align}\label{C1}
C(n)&=e^{-\frac{r}{r-q}\rho}\mathbb{E}\left[(A_{n}-K)^{+}1_{A_{n}\geq K}\right]
\\
&\leq e^{-\frac{r}{r-q}\rho}\left(\mathbb{E}[[(A_{n}-K)^{+}]^{p}]\right)^{\frac{1}{p}}
\left(\mathbb{E}[(1_{A_{n}\geq K})^{q}]\right)^{\frac{1}{q}}
\nonumber
\\
&\leq e^{-\frac{r}{r-q}\rho}\left(\mathbb{E}[(A_{n}+K)^{p}]\right)^{\frac{1}{p}}
\mathbb{P}\left(A_{n}\geq K\right)^{\frac{1}{q}}.
\nonumber
\end{align}
By Jensen's inequality, for any $x,y>0$, it is clear that for any $p\geq 2$, 
$(\frac{x+y}{2})^{p}\leq\frac{x^{p}+y^{p}}{2}$. Therefore, for any $p\geq 2$,
\begin{equation}\label{C2}
\mathbb{E}[(A_{n}+K)^{p}]
\leq 2^{p-1}(\mathbb{E}[A_{n}^{p}]+K^{p}).
\end{equation}
We can compute that
\begin{align}\label{C3}
\mathbb{E}[A_{n}^{p}]
&= n^{-p}
\mathbb{E}\left[\left(\sum_{i=1}^{n}S_{0}e^{\sigma Z_{t_{i}}+(r-q-\frac{1}{2}\sigma^{2})t_{i}}\right)^{p}\right]
\\
&= n^{-p}
\mathbb{E}\left[\left(\sum_{i=1}^{n}S_{0}e^{\sigma\sqrt{\tau}Z_{i}+(r-q-\frac{1}{2}\sigma^{2})\tau i}\right)^{p}\right]
\nonumber
\\
&\leq n^{-p}
\mathbb{E}\left[\left(\sum_{i=1}^{n}S_{0}e^{\sigma\sqrt{\tau}\max_{1\leq i\leq n}Z_{i}+|\rho|}\right)^{p}\right]
\nonumber
\\
&\leq S_{0}^{p}e^{|\rho| p}\mathbb{E}\left[e^{\frac{\sqrt{2\beta}}{n}p\max_{1\leq i\leq n}Z_{i}}\right]
\nonumber
\\
&=S_{0}^{p}e^{|\rho| p}
\mathbb{E}\left[e^{\frac{\sqrt{2\beta}}{n}p|Z_{n}|}\right]
\nonumber
\\
&=S_{0}^{p}e^{|\rho| p}
\mathbb{E}\left[e^{\frac{\sqrt{2\beta}}{\sqrt{n}}p|Z_{1}|}\right] \,,
\nonumber
\end{align}
where we used the reflection principle for the Brownian motion and the Brownian scaling property.
Note that $\mathbb{E}[e^{\theta|Z_{1}|}]$ is finite for any $\theta>0$. Hence, from \eqref{C1}, \eqref{C2}, \eqref{C3},
we conclude that for any $1<q<2$ (and thus $p>2$, where $\frac{1}{p}+\frac{1}{q}=1$),
\begin{equation}
\limsup_{n\rightarrow\infty}\frac{1}{n}\log C(n)\leq-\frac{1}{q}\lim_{n\rightarrow\infty}\frac{1}{n}\log\mathbb{P}(A_{n}\geq K)
=-\frac{1}{q}\mathcal{I}(K).
\end{equation}
Since it holds for any $1<q<2$, by letting $q\downarrow 1$, we proved \eqref{CLimit}.
\end{proof}

\subsection{In-the-Money Case}

We consider the case of in-the-money Asian options, that is 
$K>\frac{S_{0}}{\rho}(e^{\rho}-1)$ for the put option
(and $K<\frac{S_{0}}{\rho}(e^{\rho}-1)$ for the call option).
Since $A_{n}\rightarrow A_{\infty}$ a.s. we get from the bounded convergence 
theorem and put-call parity, that
$P(n)\rightarrow K-\frac{S_{0}}{\rho}(e^{\rho}-1)$
and $C(n)\rightarrow\frac{S_{0}}{\rho}(e^{\rho}-1)-K$.
The next results concern the speed of the convergence.

\begin{proposition}
When $K<\frac{S_{0}}{\rho}(e^{\rho}-1)$ and $\rho\neq 0$, 
\begin{equation}\label{CLimitII}
C(n)=e^{-\frac{r}{r-q}\rho}\left(\frac{S_{0}}{\rho}(e^{\rho}-1)-K\right)
+\frac{e^{-\frac{r\rho}{r-q}}S_{0}(e^{\rho}-1)}{2n}+O(n^{-2}).
\end{equation}
and when $K>\frac{S_{0}}{\rho}(e^{\rho}-1)$ and $\rho\neq 0$, 
\begin{equation}\label{PLimitII}
P(n)=e^{-\frac{r}{r-q}\rho}\left(K-\frac{S_{0}}{\rho}(e^{\rho}-1)\right)
-\frac{e^{-\frac{r\rho}{r-q}}S_{0}(e^{\rho}-1)}{2n}+O(n^{-2}).
\end{equation}
The case $\rho=0$ is similar. When $K<S_{0}$,
\begin{equation}
C(n)=(S_{0}-K)+e^{-n\mathcal{I}(K)+o(n)},
\end{equation}
and when $K>S_{0}$,
\begin{equation}
P(n)=(K-S_{0})+e^{-n\mathcal{I}(K)+o(n)}.
\end{equation}
\end{proposition}

\begin{proof}
When $K<\frac{S_{0}}{\rho}(e^{\rho}-1)$, we proved that $P(n)=e^{-n\mathcal{I}(k)+o(n)}$.
From put-call parity,
\begin{equation}
C(n)-P(n)=e^{-rt_{n}}\mathbb{E}[A_{n}-K]
=e^{-\frac{r}{r-q}\rho}\left[S_{0}\frac{e^{\rho}-1}{n(1-e^{-\frac{\rho}{n}})}-K\right].
\end{equation}
Therefore, 
\begin{align}
&C(n)-P(n)-e^{-\frac{r}{r-q}\rho}\left(\frac{S_{0}}{\rho}(e^{\rho}-1)-K\right)
\\
&=e^{-\frac{r\rho}{r-q}}S_{0}(e^{\rho}-1)\left[\frac{1}{n(1-e^{-\frac{\rho}{n}})}-\frac{1}{\rho}\right]
\nonumber
\\
&=e^{-\frac{r\rho}{r-q}}S_{0}(e^{\rho}-1)\left[\frac{1}{\rho-\frac{1}{2}\frac{\rho^{2}}{n}+O(n^{-2})}-\frac{1}{\rho}\right]
\nonumber
\\
&=\frac{e^{-\frac{r\rho}{r-q}}S_{0}(e^{\rho}-1)}{2n}+O(n^{-2}).
\nonumber
\end{align}
Since $P(n)=e^{-n\mathcal{I}(k)+o(n)}$, we proved \eqref{CLimitII}. Similarly, we have \eqref{PLimitII}.
\end{proof}

\subsection{At-the-Money Case}

Consider next the case of at-the-money Asian options, that is
$K=\frac{S_{0}}{\rho}(e^{\rho}-1)=A_{\infty}$. Since 
$A_{n}\rightarrow A_{\infty}$ a.s., using the bounded convergence theorem, we 
have $P(n)\rightarrow 0$ as $n\rightarrow\infty$. Put-call parity implies
that $C(n)\rightarrow 0$ as $n\rightarrow\infty$ as well. Note that in the case of out-of-the-money, we have
already seen that both $P(n)$ and $C(n)$ decay to zero exponentially fast in $n$, where the exponent
is given by $\mathcal{I}(K)$. The next result is about the speed that $P(n)$ and $C(n)$ decay to zero as $n\rightarrow\infty$
for at-the-money Asian options. We will see that, unlike the out-of-the-money Asian options, whose asymptotics
are governed by the large deviations results, the asymptotics for at-the-money case are governed
by the normal fluctuations from the central limit theorem and non-uniform 
Berry-Esseen bound.

\begin{proposition}\label{prop:ATM}
When the Asian option is at-the-money, that is, 
$K=\frac{S_{0}}{\rho}(e^{\rho}-1)=A_{\infty}$,
\begin{align}
&P(n)=e^{-\frac{r\rho}{r-q}}S_{0}\sqrt{\frac{\beta v(\rho)}{\pi}}
\frac{1}{\sqrt{n}}(1+o(1)),\\
&C(n)=e^{-\frac{r\rho}{r-q}}S_{0}\sqrt{\frac{\beta v(\rho)}{\pi}}
\frac{1}{\sqrt{n}}(1+o(1)),
\end{align}
as $n\rightarrow\infty$.
\end{proposition}

\begin{proof}
\begin{align}
C(n)&=e^{-rt_{n}}\mathbb{E}[(K-A_{n})^{+}]
=e^{-\frac{r\rho}{r-q}}\mathbb{E}\left[(A_{n}-A_{\infty})1_{A_{n}\geq A_{\infty}}\right]
\\
&=e^{-\frac{r\rho}{r-q}}S_{0}\frac{1}{\sqrt{n}}
\mathbb{E}\left[\sqrt{n}\frac{(A_{n}-A_{\infty})}{S_{0}}1_{\sqrt{n}\frac{(A_{n}-A_{\infty})}{S_{0}}\geq 0}\right].
\nonumber
\end{align}
We have proved in Proposition~\ref{prop2} that 
$\sqrt{n}\frac{(A_{n}-A_{\infty})}{S_{0}}\rightarrow N(0,2\beta v(\rho))$ as $n\rightarrow\infty$.
Intuitively, it is clear that 
$\mathbb{E}\left[\sqrt{n}\frac{(A_{n}-A_{\infty})}{S_{0}}1_{\sqrt{n}\frac{(A_{n}-A_{\infty})}{S_{0}}\geq 0}\right]\rightarrow\mathbb{E}[Z1_{Z\geq 0}]$ 
where $Z\sim N(0,2\beta v(\rho))$. 
But in order to prove this, the central limit theorem is not sufficient. 
We need a non-uniform Berry-Esseen bound \cite{Nagaev,Bikelis}, 
which we recall next.
See e.g. Pinelis \cite{Pinelis} for a survey on this subject. 

\begin{theorem}[Non-uniform Berry-Esseen bound]
For any independent and not necessarily identically 
distributed random variables $X_{1},X_{2},\ldots,X_{n}$ with zero means
and finite variances and $\mbox{Var}(W_{n})=1$, where 
$W_{n}=\sum_{i=1}^{n}X_{i}$,
let $F_{n}$ be the cumulative distribution function of $W_{n}$ and $\Phi$ the 
standard normal cumulative distribution function, 
that is $\Phi(x):=\frac{1}{\sqrt{2\pi}}\int_{-\infty}^{x}e^{-y^{2}/2}dy$.

The difference between the two distributions is bounded as \cite{Nagaev,Bikelis}
\begin{equation}
|F_{n}(x)-\Phi(x)|\leq\frac{C\sum_{i=1}^{n}\mathbb{E}|X_{i}|^{3}}{1+|x|^{3}},
\end{equation}
for any $-\infty<x<\infty$, where $C$ is a constant. The best known bound on
this constant in the general (non-identical $X_i$) case is $C<31.935$ \cite{Pinelis}. 
\end{theorem}

We have proved that
\begin{equation}\label{ThreeTerms}
\sqrt{n}\frac{(A_{n}-A_{\infty})}{S_{0}}
=\sum_{i=1}^{n}X_{i}
+\xi_{n}+\varepsilon_{n} \,,
\end{equation}
where
\begin{equation}\label{Xidef}
X_{i}:=\frac{\sqrt{2\beta}}{n^{3/2}}V_{i}\frac{e^{\frac{\rho(n+1)}{n}}-e^{\frac{\rho i}{n}}}{e^{\frac{\rho}{n}}-1},\qquad 1\leq i\leq n,
\end{equation}
where $V_{i}$ are i.i.d. $N(0,1)$ random variables and
\begin{equation}
\varepsilon_{n}:=\frac{1}{\sqrt{n}}\sum_{i=1}^{n}e^{\rho\frac{i}{n}}-\sqrt{n}\frac{e^{\rho}-1}{\rho}.
\end{equation}
The plan of the proof will be to show that the contributions from the
second and third terms in \eqref{ThreeTerms} are negligible, and to
apply the non-uniform Berry-Esseen bound to the first term in \eqref{ThreeTerms}.

From \eqref{ThreeTerms}, we have
\begin{equation}
\mathbb{E}\left[\sqrt{n}\frac{(A_{n}-A_{\infty})}{S_{0}}1_{\sqrt{n}\frac{(A_{n}-A_{\infty})}{S_{0}}\geq 0}\right]
=\mathbb{E}\left[\left(\sum_{i=1}^{n}X_{i}
+\xi_{n}+\varepsilon_{n}\right)1_{\sum_{i=1}^{n}X_{i}
+\xi_{n}+\varepsilon_{n}\geq 0}\right],
\end{equation}
which implies that
\begin{equation}\label{95}
\left|\mathbb{E}\left[\sqrt{n}\frac{(A_{n}-A_{\infty})}{S_{0}}1_{\sqrt{n}\frac{(A_{n}-A_{\infty})}{S_{0}}\geq 0}\right]
-\mathbb{E}\left[\left(\sum_{i=1}^{n}X_{i}\right)1_{\sum_{i=1}^{n}X_{i}
+\xi_{n}+\varepsilon_{n}\geq 0}\right]\right|
\leq\mathbb{E}|\xi_{n}|+|\varepsilon_{n}|.
\end{equation}
We have proved already that $\mathbb{E}|\xi_{n}|$ and $|\varepsilon_{n}|\rightarrow 0$
as $n\rightarrow\infty$. Next, notice that
\begin{align}\label{96}
&\mathbb{E}\left[\left(\sum_{i=1}^{n}X_{i}\right)1_{\sum_{i=1}^{n}X_{i}
+\xi_{n}+\varepsilon_{n}\geq 0}\right]
\\
&=\sqrt{\sum_{i=1}^{n}\mbox{Var}(X_{i})}\mathbb{E}\left[\left(\sum_{i=1}^{n}Y_{i}\right)1_{\sum_{i=1}^{n}Y_{i}
+\bar{\xi}_{n}+\bar{\varepsilon}_{n}\geq 0}\right],
\nonumber
\end{align}
where $Y_{i}=(\sum_{i=1}^{n}\mbox{Var}(X_{i}))^{-1/2}X_{i}$, $\bar{\xi}_{n}=(\sum_{i=1}^{n}\mbox{Var}(X_{i}))^{-1/2}\xi_{n}$
and $\bar{\varepsilon}_{n}=(\sum_{i=1}^{n}\mbox{Var}(X_{i}))^{-1/2}\varepsilon_{n}$, 
so that $\mbox{Var}(\sum_{i=1}^{n}Y_{i})=1$.
Recall that we already proved that
\begin{equation}
\lim_{n\rightarrow\infty}\sum_{i=1}^{n}\mbox{Var}(X_{i})=2\beta v(\rho).
\end{equation}

The expectation on the right-hand side of (\ref{96}) can be written as
\begin{align}\label{Step1}
&\mathbb{E}\left[\left(\sum_{i=1}^{n}Y_{i}\right)1_{\sum_{i=1}^{n}Y_{i}
+\bar{\xi}_{n}+\bar{\varepsilon}_{n}\geq 0}\right]
\\
&=\mathbb{E}\left[\left(\sum_{i=1}^{n}Y_{i}\right)1_{\sum_{i=1}^{n}Y_{i}
+\bar{\xi}_{n}+\bar{\varepsilon}_{n}\geq 0, |\bar{\xi}_{n}+\bar{\varepsilon}_{n}|\leq\delta}\right]
+\mathbb{E}\left[\left(\sum_{i=1}^{n}Y_{i}\right)1_{\sum_{i=1}^{n}Y_{i}
+\bar{\xi}_{n}+\bar{\varepsilon}_{n}\geq 0, |\bar{\xi}_{n}+\bar{\varepsilon}_{n}|>\delta}\right],
\nonumber
\end{align}
for any $\delta > 0$. The second term is bounded from above by the
Cauchy-Schwarz inequality
\begin{align}\label{CSI}
&\left|\mathbb{E}\left[\left(\sum_{i=1}^{n}Y_{i}\right)1_{\sum_{i=1}^{n}Y_{i}
+\bar{\xi}_{n}+\bar{\varepsilon}_{n}\geq 0, |\bar{\xi}_{n}+\bar{\varepsilon}_{n}|>\delta}\right]\right|
\\
&\leq\mathbb{E}\left[\left(\left(\sum_{i=1}^{n}Y_{i}\right)1_{\sum_{i=1}^{n}Y_{i}
+\bar{\xi}_{n}+\bar{\varepsilon}_{n}\geq 0}\right)^{2}\right]^{1/2}
\mathbb{E}[(1_{|\bar{\xi}_{n}+\bar{\varepsilon}_{n}|>\delta})^{2}]^{1/2}
\nonumber
\\
&\leq
\mathbb{E}\left[\left(\sum_{i=1}^{n}Y_{i}\right)^{2}\right]^{1/2}
\mathbb{P}(|\bar{\xi}_{n}+\bar{\varepsilon}_{n}|>\delta)^{1/2}
=\mathbb{P}(|\bar{\xi}_{n}+\bar{\varepsilon}_{n}|>\delta)^{1/2}\rightarrow 0,
\nonumber
\end{align}
as $n\rightarrow\infty$. 
The first term in (\ref{Step1}) can be written furthermore as
\begin{align}\label{2terms}
&\mathbb{E}\left[\left(\sum_{i=1}^{n}Y_{i}\right)1_{\sum_{i=1}^{n}Y_{i}
+\bar{\xi}_{n}+\bar{\varepsilon}_{n}\geq 0, |\bar{\xi}_{n}+\bar{\varepsilon}_{n}|\leq\delta}\right]
\\
&=\mathbb{E}\left[\left(\sum_{i=1}^{n}Y_{i}\right)1_{\sum_{i=1}^{n}Y_{i}
+\bar{\xi}_{n}+\bar{\varepsilon}_{n}\geq 0, |\bar{\xi}_{n}+\bar{\varepsilon}_{n}|\leq\delta,\sum_{i=1}^{n}Y_{i}\geq 0}\right]
\nonumber
\\
&\qquad\qquad
+\mathbb{E}\left[\left(\sum_{i=1}^{n}Y_{i}\right)1_{\sum_{i=1}^{n}Y_{i}
+\bar{\xi}_{n}+\bar{\varepsilon}_{n}\geq 0, 
|\bar{\xi}_{n}+\bar{\varepsilon}_{n}|\leq\delta,
\sum_{i=1}^{n}Y_{i}\leq 0}\right] \,.
\nonumber
\end{align}
The second term in \eqref{2terms} is negative and is bounded in absolute value as
\begin{align}\label{Step2}
&0 < \left|\mathbb{E}\left[\left(\sum_{i=1}^{n}Y_{i}\right)1_{\sum_{i=1}^{n}Y_{i}
+\bar{\xi}_{n}+\bar{\varepsilon}_{n}\geq 0, |\bar{\xi}_{n}+\bar{\varepsilon}_{n}|\leq\delta,\sum_{i=1}^{n}Y_{i}\leq 0}\right]
\right|
\\
&\leq
\mathbb{E}\left[\left(\sum_{i=1}^{n}Y_{i}\right)^{2}\right]^{1/2}
\mathbb{P}\left(\sum_{i=1}^{n}Y_{i}
+\bar{\xi}_{n}+\bar{\varepsilon}_{n}\geq 0, |\bar{\xi}_{n}+\bar{\varepsilon}_{n}|\leq\delta,\sum_{i=1}^{n}Y_{i}\leq 0\right)^{1/2}
\nonumber
\\
&\leq\mathbb{P}\left(-\delta\leq\sum_{i=1}^{n}Y_{i}\leq 0\right)^{1/2}
\rightarrow[\Phi(0)-\Phi(-\delta)]^{1/2} \,,
\nonumber
\end{align}
as $n\rightarrow\infty$ by the central limit theorem. 

Next, we need to estimate the first term in \eqref{2terms}. 
We first give an upper bound,
\begin{equation}\label{Step3}
\mathbb{E}\left[\left(\sum_{i=1}^{n}Y_{i}\right)1_{\sum_{i=1}^{n}Y_{i}
+\bar{\xi}_{n}+\bar{\varepsilon}_{n}\geq 0, |\bar{\xi}_{n}+\bar{\varepsilon}_{n}|\leq\delta,\sum_{i=1}^{n}Y_{i}\geq 0}\right]
\leq
\mathbb{E}\left[\left(\sum_{i=1}^{n}Y_{i}\right)1_{\sum_{i=1}^{n}Y_{i}\geq 0}\right].
\end{equation}
Next, we give a lower bound,
\begin{align}\label{Step4}
&\mathbb{E}\left[\left(\sum_{i=1}^{n}Y_{i}\right)1_{\sum_{i=1}^{n}Y_{i}
+\bar{\xi}_{n}+\bar{\varepsilon}_{n}\geq 0, |\bar{\xi}_{n}+\bar{\varepsilon}_{n}|\leq\delta,\sum_{i=1}^{n}Y_{i}\geq 0}\right]
\\
&\geq
\mathbb{E}\left[\left(\sum_{i=1}^{n}Y_{i}\right)1_{\sum_{i=1}^{n}Y_{i}\geq\delta, 
|\bar{\xi}_{n}+\bar{\varepsilon}_{n}|\leq\delta}\right]\,.
\nonumber
\end{align}
This can be written further as
\begin{align}\label{Step5}
&\mathbb{E}\left[\left(\sum_{i=1}^{n}Y_{i}\right)1_{\sum_{i=1}^{n}Y_{i}\geq\delta, 
|\bar{\xi}_{n}+\bar{\varepsilon}_{n}|\leq\delta}\right]
\\
&=\mathbb{E}\left[\left(\sum_{i=1}^{n}Y_{i}\right)1_{\sum_{i=1}^{n}Y_{i}\geq\delta}\right]
-\mathbb{E}\left[\left(\sum_{i=1}^{n}Y_{i}\right)1_{\sum_{i=1}^{n}Y_{i}\geq\delta, 
|\bar{\xi}_{n}+\bar{\varepsilon}_{n}|>\delta}\right].
\nonumber
\end{align}
By following the same argument as in \eqref{CSI}, we have
\begin{equation}\label{Step6}
\lim_{n\rightarrow\infty}
\mathbb{E}\left[\left(\sum_{i=1}^{n}Y_{i}\right)1_{\sum_{i=1}^{n}Y_{i}\geq\delta, 
|\bar{\xi}_{n}+\bar{\varepsilon}_{n}|>\delta}\right]=0.
\end{equation}
The bounds (\ref{Step3}) and (\ref{Step4}) can be combined with the bounds 
(\ref{Step2}) to obtain simpler
bounds on the expectation in (\ref{2terms}) in the $n\to \infty$ limit.
By \eqref{Step1}-\eqref{Step6}, these bounds translate into corresponding
bounds for the expectation (\ref{Step2}).
We get for any $\delta \geq 0$
\begin{align}
&\liminf_{n\rightarrow\infty}\mathbb{E}\left[\left(\sum_{i=1}^{n}Y_{i}\right)1_{\sum_{i=1}^{n}Y_{i}
+\bar{\xi}_{n}+\bar{\varepsilon}_{n}\geq 0}\right]
\\
&\geq
\liminf_{n\rightarrow\infty}\mathbb{E}\left[\left(\sum_{i=1}^{n}Y_{i}\right)1_{\sum_{i=1}^{n}Y_{i}\geq\delta}\right]
-[\Phi(0)-\Phi(-\delta)]^{1/2},
\nonumber
\end{align}
and
\begin{align}
\limsup_{n\rightarrow\infty}\mathbb{E}\left[\left(\sum_{i=1}^{n}Y_{i}\right)1_{\sum_{i=1}^{n}Y_{i}
+\bar{\xi}_{n}+\bar{\varepsilon}_{n}\geq 0}\right]
\leq
\limsup_{n\rightarrow\infty}\mathbb{E}\left[\left(\sum_{i=1}^{n}Y_{i}\right)1_{\sum_{i=1}^{n}Y_{i}\geq 0}\right]
\,.
\end{align}

Finally, take the $\delta\to 0$ limit, which gives
\begin{eqnarray}
\lim_{n\to \infty} 
\mathbb{E}\left[\left(\sum_{i=1}^{n}Y_{i}\right)1_{\sum_{i=1}^{n}Y_{i}
+\bar{\xi}_{n}+\bar{\varepsilon}_{n}\geq 0}\right]
= 
\lim_{n\to \infty} 
\mathbb{E}\left[\left(\sum_{i=1}^{n}Y_{i}\right)1_{\sum_{i=1}^{n}Y_{i}\geq 0}
\right] \,.
\end{eqnarray} 
The non-uniform Berry-Esseen bound can be applied to compute the expectation
on the right-hand side.

The sums of third moments appearing in the non-uniform Berry-Esseen bound
are estimated as follows.
Recalling that $Y_{i} = (\sum_{i=1}^n \mbox{Var}(X_i))^{-1/2} X_i$ 
where $X_i$ are defined in terms of $N(0,1)$ i.i.d. random variables $V_i$ as 
given in (\ref{Xidef}), we find
\begin{align}
\sum_{i=1}^{n}\mathbb{E}|Y_{i}|^{3}
&=\frac{1}{\left(\sum_{i=1}^{n}\mbox{Var}(X_{i})\right)^{3/2}}
\sum_{i=1}^{n}\mathbb{E}|X_{i}|^{3}
\\
&=\frac{1}{\left(\sum_{i=1}^{n}\mbox{Var}(X_{i})\right)^{3/2}}
\frac{(2\beta)^{3/2}}{n^{1/2}}\mathbb{E}|V_{1}|^{3}
\sum_{i=1}^{n}\left(\frac{e^{\frac{\rho(n+1)}{n}}-e^{\frac{\rho i}{n}}}{n(e^{\frac{\rho}{n}}-1)}\right)^{3}\frac{1}{n}
\nonumber
\\
&\leq\frac{C_{0}(\rho)}{n^{1/2}},
\nonumber
\end{align}
where $C_{0}(\rho)>0$ depends only on $\rho$. 
Therefore, by the non-uniform Berry-Esseen bound, we have
\begin{equation}
|F_{n}(x)-\Phi(x)|\leq\frac{C_{1}(\rho)}{n^{1/2}}\frac{1}{1+|x|^{3}},
\end{equation}
for any $-\infty<x<\infty$, where $F_{n}$ is the cumulative distribution 
function of $\sum_{i=1}^{n}Y_{n}$, and $C_1(\rho)>0$ is another constant.
Hence, we have, with $Z \sim N(0,2\beta v(\rho))$
\begin{align}
\left|\mathbb{E}\left[\left(\sum_{i=1}^{n}Y_{i}\right)1_{\sum_{i=1}^{n}Y_{i}\geq 0}\right]
-\mathbb{E}[Z1_{Z\geq 0}]\right|
&=\left|\int_{0}^{\infty}xdF_{n}(x)-\int_{0}^{\infty}xd\Phi(x)\right|
\\
&=\left|\int_{0}^{\infty}F_{n}(x)dx-\int_{0}^{\infty}\Phi(x)dx\right|
\nonumber
\\
&\leq\int_{0}^{\infty}\frac{C_{1}(\rho)}{n^{1/2}}\frac{1}{1+|x|^{3}}dx.
\nonumber
\end{align}
which goes to zero as $n\rightarrow\infty$. 
We conclude that we have
\begin{equation}
C(n)=e^{-\frac{r\rho}{r-q}}S_{0}\mathbb{E}[Z1_{Z\geq 0}]\frac{1}{\sqrt{n}}(1+o(1)),
\end{equation}
as $n\rightarrow\infty$, where $Z\sim N(0,2\beta v(\rho))$.
The expectation is given explicitly by
\begin{equation}
\mathbb{E}[Z1_{Z\geq 0}]=\sqrt{2\beta v(\rho)}\frac{1}{\sqrt{2\pi}}\int_{0}^{\infty}xe^{-\frac{x^{2}}{2}}dx
=\sqrt{\frac{\beta v(\rho)}{\pi}} \,.
\end{equation}
This completes the proof of the asymptotics for the at-the-money call 
option $C(n)$.
The asymptotics for the price of the at-the-money put option $P(n)$
can be obtained by using put-call parity. The proof is complete. 
\end{proof}

\section{Asymptotics for Floating Strike Asian Options}
\label{sec:5}

We consider in this section the floating strike Asian options, which 
are a variation of the standard Asian option.
The floating strike Asian call option with strike $K$
and weight $\kappa$ has payoff $(\kappa S_{T}-A_{T})^{+}$ at maturity $T$
and the floating strike put option has payoff $(A_{T}-\kappa S_{T})^{+}$ at 
maturity $T$.

The floating-strike Asian option is more difficult to price than the fixed-strike case
because the joint law of $S_{T}$ and $A_{T}$ is needed. Also, the 
one-dimensional PDE that the floating-strike Asian price satisfies after a 
change of num\'{e}raire is difficult to solve numerically as the Dirac delta 
function appears as a coefficient, see e.g. \cite{Ingersoll}, \cite{RogersShi}, 
\cite{Alziary}. See \cite{Ritchken,Chung,HW} for alternative methods which
have been proposed to deal with this problem.

It has been shown by Henderson and Wojakowski \cite{HW} that the floating-strike 
Asian options with continuous time averaging can be related to fixed strike 
ones. 
These equivalence relations have been extended to discrete time averaging Asian
options in \cite{VDLDG}. According to these relations we have
\begin{align}
\label{equiv1}
&e^{-rt_{n}}\mathbb{E}\left[(\kappa S_{t_{n}}-A_{n})^{+}\right] =
e^{-qt_{n}}\mathbb{E}_*\left[(\kappa S_{0}-A_{n})^{+}\right] \,, \\
\label{equiv2}
&e^{-rt_{n}}\mathbb{E}\left[(A_{n}-\kappa S_{t_{n}})^{+}\right] =
e^{-qt_{n}}\mathbb{E}_*\left[(A_{n}-\kappa S_{0})^{+}\right] \,,
\end{align}
The expectations on the right-hand side are taken with respect to a different
measure $\mathbb{Q}_*$, where the asset price $S_t$ follows  the process
\begin{eqnarray}
dS_t = (q-r) S_t dt + \sigma S_t dW_t^*\,,
\end{eqnarray}
with $W_t^*$ a standard Brownian motion in the $\mathbb{Q}_*$ measure.

We are interested in the asymptotics of the price of the Asian call/put options
with payoffs $(\kappa S_{t_{n}}-A_{n})^{+}$ and $(A_{n}-\kappa S_{t_{n}})^{+}$,
\begin{align}
&C(n):=e^{-rt_{n}}\mathbb{E}\left[(\kappa S_{t_{n}}-A_{n})^{+}\right],
\\
&P(n):=e^{-rt_{n}}\mathbb{E}\left[(A_{n}-\kappa S_{t_{n}})^{+}\right].
\end{align}

As $n\rightarrow\infty$, $\kappa S_{t_{n}}-A_{n}\rightarrow 
\kappa S_{0}e^{\rho}-S_{0}\frac{e^{\rho}-1}{\rho}$ a.s.
When $\kappa < \frac{1}{\rho}(1-e^{-\rho})$ the call option is 
out-of-the-money and the put option is in-the-money. 
When $\kappa > \frac{1}{\rho}(1-e^{-\rho})$, the call option is in-the-money
and the put option is out-of-the-money. 
When $\kappa=\frac{1}{\rho}(1-e^{-\rho})$, the call and put options
are at-the-money. 

For the expectations on the right-hand side of the equivalence relations 
(\ref{equiv1}), (\ref{equiv2}) we have that as $n\to \infty$,
$\kappa S_0 - A_n \to \kappa S_0 - S_0 \frac{e^{-\rho}-1}{-\rho}$ a.s.
We conclude that for $\kappa < \frac{1}{\rho}(1-e^{-\rho})$ these
equivalence relations map an out-of-money
floating strike call (put) Asian option onto an out-of-money fixed 
strike put (call) Asian option. For $\kappa > \frac{1}{\rho}(1-e^{-\rho})$ a
similar relation holds between the respective in-the-money Asian options.

Let us derive the asymptotics of the price of the floating strike Asian options.
This could be expressed in terms of the asymptotics of the fixed strike Asian options
obtained in the previous sections, with the help of the equivalence relations.
An alternative way is to derive directly the large deviation result for the 
floating strike Asian options. Then we will relate the rate function to that 
for the fixed strike Asian options, and show that this is consistent with the
equivalence relations.

We have the following result for the asymptotics of floating strike
Asian options.

\begin{proposition}\label{FloatProp}
(i) When $\kappa < \frac{1}{\rho}(1-e^{-\rho})$, the call option is 
out-of-the-money,
\begin{equation}
C(n)=e^{-n\mathcal{H}(0)+o(n)},\qquad\text{as $n\rightarrow\infty$}\,,
\end{equation}
and the put option is in-the-money,
\begin{equation}
P(n)=
\begin{cases}
-\kappa S_{0}e^{-\frac{r}{r-q}\rho}+ S_{0}e^{-\frac{r}{r-q}\rho}\frac{e^{\rho}-1}{\rho}
+\frac{ e^{-\frac{r\rho}{r-q}}S_{0}(e^{\rho}-1)}{2n}+O(n^{-2}) &\rho\neq 0,
\\
(1-\kappa)S_{0}+e^{-n\mathcal{H}(0)+o(n)} &\rho=0.
\end{cases}
\end{equation}

(ii) When $\kappa>\frac{1}{\rho}(1-e^{-\rho})$, the put option is 
out-of-the-money
\begin{equation}
P(n)=e^{-n\mathcal{H}(0)+o(n)},\qquad\text{as $n\rightarrow\infty$},
\end{equation}
and the call option is in-the-money,
\begin{equation}
C(n)=
\begin{cases}
\kappa S_{0}e^{-\frac{r}{r-q}\rho}- S_{0}e^{-\frac{r}{r-q}\rho}\frac{e^{\rho}-1}{\rho}
-\frac{ e^{-\frac{r\rho}{r-q}}S_{0}(e^{\rho}-1)}{2n}+O(n^{-2}) &\rho\neq 0,
\\
S_{0}(\kappa - 1)+e^{-n\mathcal{H}(0)+o(n)} &\rho=0.
\end{cases}
\end{equation}

The rate function in (i) and (ii) is given by 
\begin{equation}\label{HZero}
\mathcal{H}(0):=\inf_{g\in\mathcal{AC}_{0}[0,1],\kappa e^{\sqrt{2\beta}g(1)}-\int_{0}^{1}e^{\sqrt{2\beta}g(y)}dy=0}
\frac{1}{2}\int_{0}^{1}\left(g'(x)-\frac{\rho}{\sqrt{2\beta}}\right)^{2}dx \,.
\end{equation}

(iii) When $\kappa=\frac{1}{\rho}(1-e^{-\rho})$, the call and put options are in-the-money,
\begin{equation}
\lim_{n\rightarrow\infty}\sqrt{n}C(n)
=\lim_{n\rightarrow\infty}\sqrt{n}P(n)
=S_{0}e^{-\frac{r\rho}{r-q}}\mathbb{E}[Z1_{Z\geq 0}],
\end{equation}
where $Z=N(0,s)$ is a normal random variable with mean $0$ and variance
\begin{equation}
s^2 = 
\frac{2\beta}{\rho^{2}}\left[1-\frac{2}{\rho}(e^{\rho}-1)+\frac{e^{2\rho}-1}{2\rho}\right].
\end{equation}
\end{proposition}

\begin{proof}
The proof is similar to the fixed-strike case. 
The sketch of the proof will be given in the Appendix.
\end{proof}

We show next that the rate function $\mathcal{H}(0)$ can be simply related
to $\mathcal{I}(x)$
defined in \eqref{RateFunction}.
Recall that we showed explicitly the dependence of $\rho$ of the respective
rate functions $H(\cdot)$ and $I(\cdot)$. Abusing the notations a bit to emphasize the dependence on $\rho$, 
let $H(\cdot;\rho):=H(\cdot)$ and $I(\cdot;\rho):=I(\cdot)$.
We have the following result, which is clearly consistent with the equivalence
relations (\ref{equiv1}), (\ref{equiv2}).

\begin{proposition}
The rate functions for the fixed strike and floating strike Asian options
are related as
\begin{equation}\label{H0Iequiv}
\mathcal{H}(0;\rho) = \mathcal{I}(\kappa S_0; - \rho)\,.
\end{equation}
\end{proposition}

\begin{proof}
The functionals in the variational problems for $\mathcal{H}(0)$ and
$\mathcal{I}(x)$ are identical, and the only difference is in the 
constraints on $g(x)$. The constraints can be related as follows.

Let us express $g(x)$ in the variational problem for $\mathcal{H}(0)$  in terms
of a new function $h(x)$ defined as $g(x) = g(1) + h(1-x)$. This function
satisfies the constraint $h(0)=0$. The rate function is now given by
\begin{equation}\label{HZeroh}
\mathcal{H}(0):=\inf_{h\in\mathcal{AC}_{0}[0,1],\kappa-\int_{0}^{1}e^{\sqrt{2\beta}h(y)}dy=0}
\frac{1}{2}\int_{0}^{1}\left(h'(x)+\frac{\rho}{\sqrt{2\beta}}\right)^{2}dx,
\end{equation}
It is easy to see that this variational problem is identical to that for the
rate function $\mathcal{I}(x)$, identifying $K/S_0=\kappa$ and $\rho \to -\rho$.
This concludes the proof of the relation (\ref{H0Iequiv}).
\end{proof}

\section{Implied volatility and numerical tests}
\label{sec:6}

It has become accepted market practice to quote European option prices
in terms of their implied volatility. This is defined as that value of the
log-normal volatility which, upon substitution into the Black-Scholes formula,
reproduces the market option prices. A similar normal implied volatility 
can be defined in terms of the Bachelier formula. 

Although Asian options are quoted in practice by price, and not by implied
volatility, it is convenient to define an equivalent 
implied volatility also for these options. We will define the equivalent
log-normal implied volatility of an Asian option with strike $K$ and 
maturity $T$ as that value of the volatility $\Sigma_{\rm LN}(K,T)$
which reproduces the Asian option
price when substituted into the Black-Scholes formula for an European 
option with the same parameters $(K,T)$
\begin{align}\label{CAsian}
C(K,S_0,T) &= e^{-rT} (A_\infty \Phi(d_1) - K \Phi(d_2)), \\
P(K,S_0,T) &= e^{-rT} (K \Phi(-d_2) - A_\infty \Phi(-d_1)) ,\nonumber
\end{align}
where
\begin{equation}
A_\infty = S_0 \frac{1}{\rho} (e^\rho-1) = 
\frac{1}{(r-q)T} (e^{(r-q)T}-1),
\end{equation}
and
\begin{equation}
d_{1,2} = \frac{1}{\Sigma_{\rm LN}(K,T)\sqrt{T}}\left(
\log\frac{A_\infty}{K} \pm \frac12 \Sigma_{\rm LN}^2(K,T) T \right).
\end{equation}
The equivalent log-normal volatility $\Sigma_{\rm LN}$ defined in this way
exists for any Asian option call price $C(K,S_0,T)$ satisfying the
Merton bounds $(A_\infty - K)^+ \leq e^{rT} C(K,S_0,T) \leq A_\infty$ 
\cite{RoperRutk}. For finite $n$ the price of the Asian option is
bounded as\footnote{The lower bound follows from the convexity of the
payoff $(x-K)^+$ and the upper bound follows from $(x-K)^+ \leq x$.} 
$(\mathbb{E}[A_n] - K)^+ \leq e^{rT} C(K,S_0,T) \leq \mathbb{E}[A_n]$
with $\mathbb{E}[A_n]=\frac{1}{n} \frac{e^\rho-1}{1-e^{-\rho/n}}$, 
so the required bounds are satisfied for $n\to \infty$. 

One can define also a normal
equivalent volatility  $\Sigma_{\rm N}(K,T)$ of an Asian option, as that
volatility which reproduces the Asian option price when substituted into the
Bachelier option pricing formula. 

We would like to study the implications of the asymptotic results for 
Asian option prices derived in  Section~\ref{sec:4} for the equivalent
log-normal volatility $\Sigma_{\rm LN}$, and for the equivalent normal 
volatility $\Sigma_{\rm N}$. 
This is given by the following result.

\begin{proposition}\label{prop:impvol}
i) The asymptotic normal and log-normal equivalent implied volatilities of an 
OTM Asian option in the $n\to \infty$ limit at constant $\beta = \frac12
\sigma^2 t_n n$ are given by
\begin{align}\label{SigBS}
& \lim_{n\to \infty} \frac{\Sigma_{\rm LN}^2(K,n)}{\sigma^2} = 
\frac12 \frac{\log^2(K/A_\infty)}
{\mathcal{J}(K/S_0,\rho)}\,, \\
\label{SigN}
& \lim_{n\to \infty} \frac{\Sigma_N^2(K,n)}{\sigma^2} = 
\frac12 \frac{(K - A_\infty)^2}
{\mathcal{J}(K/S_0, \rho)}\,,
\end{align}
where $\mathcal{J}(K/S_0,\rho)$ is related to the rate function 
$\mathcal{I}(x)$ as in (\ref{Jcaldef}), and is given by 
Proposition~\ref{prop:rhogeneral}.

ii) The equivalent log-normal implied volatility for $n\to \infty$ of an 
at-the-money Asian option is
\begin{equation}\label{ATM}
\lim_{n\to \infty} \frac{\Sigma_{\rm LN}(A_\infty,n)}{\sigma} = 
\frac{S_0}{A_\infty} \sqrt{v(\rho)},
\end{equation}
and the corresponding result for the equivalent normal implied volatility is
\begin{equation}\label{ATMN}
\lim_{n\to \infty} \frac{\Sigma_{\rm N}(A_\infty,n)}{\sigma} = 
S_0 \sqrt{v(\rho)} \,.
\end{equation}
\end{proposition}

\begin{proof}
The proof is given in the Appendix.
\end{proof}

We note that in (\ref{SigBS}) $\sigma$ depends implicitly on $n$ as the
limit is taken at fixed $\beta$. In particular, in the fixed maturity
regime $\tau n = T$ fixed, we have $\sigma \sim n^{-1/2}$
such that both $\sigma$ and $\Sigma_{\rm LN}(K)$
approach 0 as $n\to \infty$, in such a way that their ratio approaches
a finite non-zero value. We will use this relation for finite $n$
to approximate the 
equivalent log-normal implied volatility $\Sigma_{\rm LN}(K)$ as
\begin{equation}
\Sigma^2_{\rm LN}(K,n) = \sigma^2 \frac12 \frac{\log^2(K/A_\infty)}
{\mathcal{J}(K/S_0,\rho)}\,.
\end{equation}
and analogously for $\Sigma_{\rm N}(K)$. These volatilities can be used
together with (\ref{CAsian}) to obtain approximations for Asian option prices.

We show in Table~\ref{Table:1} numerical results for the asymptotic 
approximation for the Asian options obtained from (\ref{CAsian}), for a
few scenarios proposed in \cite{FMW}.
They are compared against a few alternative methods considered in the
literature: the method of Linetsky \cite{Linetsky}, PDE methods 
\cite{FPP2013,VecerRisk}, inversion of Laplace transform \cite{DS,Shaw}, 
and the log-normal approximation \cite{Levy}
corresponding to continuous-time averaging. 

The numerical agreement of the asymptotic result with the precise results of 
the spectral expansion \cite{Linetsky} is very good, and the difference is 
always below $0.5\%$ in relative value. A more appropiate test compares the
difference to the option Vega $\mathcal{V}$: the approximation error of the
asymptotic result is always below $0.24\mathcal{V}$ (compared with the log-normal
approximation which has an error as large as $1.54\mathcal{V}$ (for scenario 7)).
This is smaller than the typical precision on $\sigma$ around the ATM point,
and compares well with typical bid-ask spreads for
Asian options which can be $\sim 1\mathcal{V}$ for maturities up to 1-2Y.

\begin{remark}
We comment on the relation of the asymptotic implied volatility (\ref{SigBS})
to the log-normal approximation \cite{Levy}.
The log-normal approximation \cite{Levy} corresponds to a flat equivalent 
log-normal volatility 
$\Sigma_{\rm LN}^{(Levy)}(T)$. In contrast, the asymptotic equivalent log-normal 
implied volatility $\Sigma_{\rm LN}(K)$ given by (\ref{SigBS}) has a 
non-trivial dependence on strike. It can be easily shown that the
log-normal implied volatility reproduces the asymptotic equivalent
implied volatility at the ATM point in the limit 
$\lim_{\sigma^2 T\to 0,r T=\rho} \Sigma_{\rm LN}^{(Levy)}(T) = 
\Sigma_{\rm LN}(K=A_\infty)$. 
\end{remark}

The results of Table~\ref{Table:1} show that the
asymptotic result is an improvement over the log-normal approximation.

\begin{remark}
The results of \cite{FMW,Linetsky} are obtained using continuous-time
averaging, while our result (\ref{SigBS}) was derived for discrete time
Asian options. However, we note that the result (\ref{SigBS}) does not
depend on the size of the time step $\tau$, so it should hold for arbitrarily
small time step. It is shown elsewhere \cite{SMAO} that a result similar to
(\ref{SigBS}) holds for the small maturity limit of continuous time Asian 
options at fixed $\sigma,r,q$, with the substitution $\rho=0$. The limiting 
procedure adopted in this paper, of taking $n\to \infty$
at fixed $\beta,\rho$, allows one to take into account the dependence on $r,q$ 
in the short maturity expansion.
\end{remark}

\begin{table}[t!]
\caption{\label{Table:0} 
The 7 benchmark scenarios considered for pricing Asian options 
in \cite{FMW,Linetsky}, etc. Here, $q=0$.}
\begin{center}
\begin{tabular}{c|ccccc}
\hline
Scenario & $r$ & $T$ & $S_0$ & $K$ & $\sigma$ \\
\hline
\hline
1 & 0.02   & 1 & 2    & 2 & 0.1  \\
2 & 0.18   & 1 & 2    & 2 & 0.3  \\
3 & 0.0125 & 2 & 2    & 2 & 0.25  \\
4 & 0.05   & 1 & 1.9  & 2 & 0.5  \\
5 & 0.05   & 1 & 2    & 2 & 0.5   \\
6 & 0.05   & 1 & 2.1  & 2 & 0.5  \\
7 & 0.05   & 2 & 2    & 2 & 0.5   \\
\hline
\end{tabular}
\end{center}
\end{table}

\begin{table}[t!]
\caption{\label{Table:1} 
Numerical results for Asian call options under the 7 scenarios considered
in \cite{FMW,Linetsky}, etc. 
FPP3: the 3rd order approximations in Foschi et al. \cite{FPP2013},
Vecer: the PDE method from \cite{VecerRisk},
MAE3: the matched asymptotic expansions 
from Dewynne and Shaw \cite{DS}, Mellin500: the 
Mellin transform based method in Shaw \cite{Shaw}.
The last column shows the results from the
spectral expansion in \cite{Linetsky}, and the LN column shows the
result of the log-normal approximation \cite{Levy}. The column PZ
gives the results of the asymptotic result of this paper using (\ref{CAsian}). }
\begin{center}
\begin{tabular}{c|ccccccc}
\hline
Scenario & FPP3 & MAE3 & Mellin500 & Vecer & PZ & LN & Linetsky\\
\hline
\hline
1 & 0.055986 & 0.055986 & 0.056036 & 0.055986 & 0.055998 & 0.056054 & 0.055986 \\
2 & 0.218387 & 0.218369 & 0.218360 & 0.218388 & 0.218480 & 0.219829 & 0.218387 \\
3 & 0.172267 & 0.172263 & 0.172369 & 0.172269 & 0.172460 & 0.173490 & 0.172269 \\
4 & 0.193164 & 0.193188 & 0.192972 & 0.193174 & 0.193692 & 0.195379 & 0.193174 \\
5 & 0.246406 & 0.246382 & 0.246519 & 0.246416 & 0.246944 & 0.249791 & 0.246416 \\
6 & 0.306210 & 0.306139 & 0.306497 & 0.306220 & 0.306744 & 0.310646 & 0.306220 \\
7 & 0.350040 & 0.349909 & 0.348926 & 0.350095 & 0.351517 & 0.359204 & 0.350095 \\
\hline
\end{tabular}
\end{center}
\end{table}



In order to address the performance of the asymptotic results in the
small volatility and maturity regime we compare our results against those in
Table 4 of \cite{FPP2013}. As pointed out in \cite{Shaw,FMW}, some  of the
methods proposed  in the literature have numerical issues in these regimes of
the model parameters. 
The scenarios considered for this test correspond to $\sigma = 0.01, S_0=100, 
r=0.05, q=0$, and three choices of maturity and strike as shown in 
Table~\ref{Table:SmallVol}. For reasons of space economy, we present only 
a subset of the test results in Table 4 of \cite{FPP2013}, which show the best
agreement with a Monte Carlo calculation. The asymptotic results are in very
good agreement with the alternative methods shown. We note that the computing
time required by the asymptotic method is very good, as it requires only the
solution of a simple non-linear algebraic equation, and the evaluation of a 
function. 

\begin{table}
\caption{\label{Table:SmallVol} 
Test results for Asian call options under small volatility $\sigma=0.01,
S_0=100, r=0.05, q=0$.
The column FPP3 shows the 3rd order approximations in Foschi et al. \cite{FPP2013}.
The column MAE3 shows the results using the matched asymptotic expansions 
from Dewynne and Shaw \cite{DS}. The column Mellin500 shows the results of the 
Mellin transform based method in Shaw \cite{Shaw}. The column PZ shows the 
asymptotic results of this paper.
}
\begin{center}
\begin{tabular}{cc|c|ccc}
\hline
$T$ & $K$ & PZ & FPP3 & MAE3 & Mellin500 \\
\hline
\hline
0.25 & 99  & 1.60739   & $1.60739 \times 10^0$ & $1.60739 \times 10^0$ & $1.51718 \times 10^0$ \\
0.25 & 100 & 0.621359  & $6.21359 \times 10^{-1}$ & $6.21359 \times 10^{-1}$ & $6.96855 \times 10^{-1}$ \\
0.25 & 101 & 0.0137615 & $1.37618 \times 10^{-2}$ & $1.37615 \times 10^{-2}$ & $1.60361 \times 10^{-2}$ \\
\hline
1.00 & 97   & 5.2719   & $5.27190 \times 10^0$ & $5.27190 \times 10^0$ & $5.27474 \times 10^0$ \\
1.00 & 100  & 2.41821   & $2.41821 \times 10^0$ & $2.41821 \times 10^0$ & $2.43303 \times 10^0$ \\
1.00 & 103  & 0.0724339 & $7.26910 \times 10^{-2}$ & $7.24337 \times 10^{-2}$ & $8.50816 \times 10^{-2}$ \\
\hline
5.00 & 80   & 26.1756   & $2.61756 \times 10^1$ & $2.61756 \times 10^1$ & $2.61756 \times 10^1$ \\
5.00 & 100  & 10.5996   & $1.05996 \times 10^1$ & $1.05996 \times 10^1$ & $1.05993 \times 10^1$ \\
5.00 & 120  & $5.8331\cdot 10^{-6}$ & $2.06699 \times 10^{-5}$ & $5.73317 \times 10^{-6}$ & $1.42235 \times 10^{-3}$ \\
\hline
\end{tabular}
\end{center}
\end{table}


We present in Table~\ref{Table:discreteAsian} a comparison with the 
test results for discretely sampled 
Asian options corresponding to the scenarios considered in Table B of 
\cite{VecerRisk}. These scenarios have parameters
$r = 0.1\,, q = 0\,, \sigma = 0.4\,,T = 1\,,  K = 100$.
The results are compared against those obtained in 
\cite{VecerRisk,TavellaRandall,Curran}. The asymptotic results agree with
the alternative methods up to about 1\%-1.5\% in relative error.

\begin{table}[t!]
\caption{\label{Table:discreteAsian} 
Asymptotic results for discretely sampled Asian call options under the 
scenarios considered in Table B of \cite{VecerRisk}, comparing with the
results of \cite{VecerRisk,TavellaRandall,Curran}.}
\begin{center}
\begin{tabular}{c|ccc}
\hline
 & $S_0 = 95$ & $S_0=100$ & $S_0=105$ \\
\hline
\hline
\multicolumn{1}{c|}{} & \multicolumn{3}{c}{Vecer}  \\
\hline
$n=250$  & 8.4001 & 11.1600 & 14.3073 \\
$n=500$  & 8.3826 & 11.1416 & 14.2881 \\
$n=1000$ & 8.3741 & 11.1322 & 14.2786 \\
$\infty$ & 8.3661 & 11.1233 & 14.2696 \\
\hline
\multicolumn{1}{c|}{} & \multicolumn{3}{c}{Tavella-Randall}  \\
\hline
$n=250$  & 8.3972 & 11.1573 & 14.3054 \\
$n=500$  & 8.3804 & 11.1392 & 14.2866 \\
$n=1000$ & 8.3719 & 11.1300 & 14.2771 \\
$\infty$ & 8.3640 & 11.1215 & 14.2681 \\
\hline
\multicolumn{1}{c|}{} & \multicolumn{3}{c}{Curran}  \\
\hline
$n=250$  & 8.3972 & 11.1572 & 14.3048 \\
$n=500$  & 8.3801 & 11.1388 & 14.2857 \\
$n=1000$ & 8.3715 & 11.1296 & 14.2762 \\
$\infty$ & $-$ & $-$ & $-$ \\
\hline
PZ             &  8.3789  & 11.1362 & 14.2818 \\
\hline
\end{tabular}
\end{center}
\end{table}


Finally, in order to test the asymptotic relation (\ref{SigBS}) for the
equivalent log-normal implied volatility we show in Figure~\ref{Fig:sigLN} 
the equivalent log-normal implied volatility of several Asian options obtained 
by numerical simulation (black dots). These results
are obtained by Monte Carlo pricing of Asian options with parameters
\begin{eqnarray}\label{scenario}
\sigma=0.2\,,\qquad r=q=0\,,\qquad \tau=0.01 ,
\end{eqnarray}
and $n=50,100,200$ averaging dates. 
The Monte Carlo calculation used $N_{\rm MC}=10^6$ samples. The strikes 
considered cover a region around the 
ATM point $K=S_0$; the numerical precision of the simulation decreases 
rapidly outside of this region. We note very good agreement with the 
asymptotic result of Proposition~\ref{prop:impvol}, even for $n$ as low as 50.

\section*{Acknowledgements}

We are grateful to an anonymous referee and the editor for their helpful comments and suggestions.
D.~P. would like to thank Dyutiman Das and Roussen Roussev for useful 
discussions about Asian options in financial practice.
L.~Z. is partially supported by NSF Grant DMS-1613164.

\section{Appendix}

\begin{proof}[Proof of Proposition \ref{VarProblem}]

The variational problem appearing in equation (\ref{varprob40}) can be
written equivalently by introducing the  function $f(x) = bg(x)$ as
\begin{equation}\label{varprob40p}
\lambda(a,b;\rho) = \frac{1}{b^2} \sup_{f\in\mathcal{AC}_{0}[0,1]}
\left\{ - ab^2 \int_{0}^{1}e^{f(x)}dx
-\frac{1}{2}\int_{0}^{1}\left(f'(x)-\rho \right)^{2}dx
\right\}\,.
\end{equation}

The functional $\Lambda[f]$ appearing in this variational problem can be
rewritten as
\begin{align}\label{Lambdadef}
\Lambda[f] &= -ab^2 \int_0^1 dx e^{f(x)} - \frac12 \int_0^1
\left(f'(x) - \rho \right)^2 dx   \\
&= -a b^2 \int_0^1  e^{f(x)} dx- \frac12 \int_0^1
[f'(x)]^2 dx + f(1) \rho - \frac12 \rho^2 \,. \nonumber
\end{align}
In the second line we integrated by parts and wrote $\int_0^1 f'(x) dx  = f(1)$
where we took into account the constraint $f(0)=0$.
Although in Proposition~\ref{VarProblem} we have $a>0$, the 
variational problems in Section~\ref{sec:4} require also the case
of negative $a$. For this reason we will treat here both cases of positive 
and negative $a$. 

The optimal function $f(x)$ satisfies the Euler-Lagrange equation 
\begin{equation}\label{ODE}
f''(x) = a b^2 e^{f(x)},
\end{equation}
with the boundary conditions
\begin{equation}\label{bc}
f(0) = 0\,,\qquad f'(1) = \rho \,.
\end{equation}
The second boundary condition (at $x=1$) is a transversality condition.

We observe that the quantity
\begin{equation}\label{Edef}
E = - ab^2 e^{f(x)} + \frac12 [f'(x)]^2 = -ab^2 e^{f(1)} + \frac12 \rho^2
\end{equation}
is a constant of motion of the differential equation (\ref{ODE}). Its value 
was expressed in terms of $f(1)$ by taking $x=1$ and using the boundary 
condition (\ref{bc}). Taking the integral of this relation over $x:(0,1)$ can be
used to eliminate the integral of $\frac12 [f'(x)]^2$ in the functional
$\Lambda[f]$. This can be put into the equivalent form
\begin{align}\label{Lambda2}
\Lambda[f]  = -2 a b^2 \int_0^1  e^{f(x)} dx + ab^2 e^{f(1)}
 + f(1) \rho -  \rho^2 \,. 
\end{align}

The Euler-Lagrange equation (\ref{ODE}) can be solved exactly. Two
independent solutions of this equation are
\begin{align}\label{f1}
& f_1(x) = \delta x - 2 \log\left( \frac{e^{\delta x} + \gamma}{1+\gamma} \right), \\
\label{f2}
& f_2(x) = - 2 \log |\cos (\xi x + \eta) | + 2\log |\cos \eta |\,.
\end{align}
The first solution was given in \cite{GR} where a related differential equation
appears in the context of optimal sampling for Monte Carlo pricing of Asian
options.
It is easy to see by direct substitution into (\ref{ODE}) that these functions
satisfy this equation, with the appropriate boundary condition at $x=0$.
Requiring that the coefficient in this equation and the boundary condition 
$f'(1)=\rho$ are satisfied gives two conditions. 

For $f_1(x)$ we have the conditions
\begin{align}\label{f1const}
2\gamma \delta^2 = - a b^2 (1 + \gamma)^2 ,\qquad
\delta \frac{\gamma - e^\delta}{\gamma + e^\delta} = \rho \,.
\end{align}
Eliminating $\gamma$ between these two equations as
$\gamma = \frac{\delta + \rho}{\delta-\rho}e^\delta$ 
gives an equation for $\delta$:
\begin{equation}\label{deltaeq}
\delta^2 - \rho^2 = - 2ab^2 \left( \cosh\left(\frac12\delta\right) + \frac{\rho}{\delta}
\sinh\left(\frac12\delta\right) \right)^2\,.
\end{equation}

For $f_2(x)$ we obtain the conditions
\begin{align}\label{f2const}
2\xi^2 = ab^2 \cos^2\eta \,,\qquad
2\xi \tan(\xi + \eta) = \rho \,.
\end{align}
The second relation allows one to eliminate $\eta$ as
\begin{equation}
\tan \eta = \frac{\frac12 \rho - \xi \tan\xi}
{\xi + \frac12\rho \tan\xi} \,.
\end{equation}
We obtain the equation for $\xi$
\begin{equation}\label{lambdaeq}
2\xi^2 (4\xi^2 + \rho^2) = ab^2 
(2\xi\cos\xi + \rho \sin\xi)^2 \,.
\end{equation}

\begin{figure}[t!]
\begin{center}
\includegraphics[height=25mm]{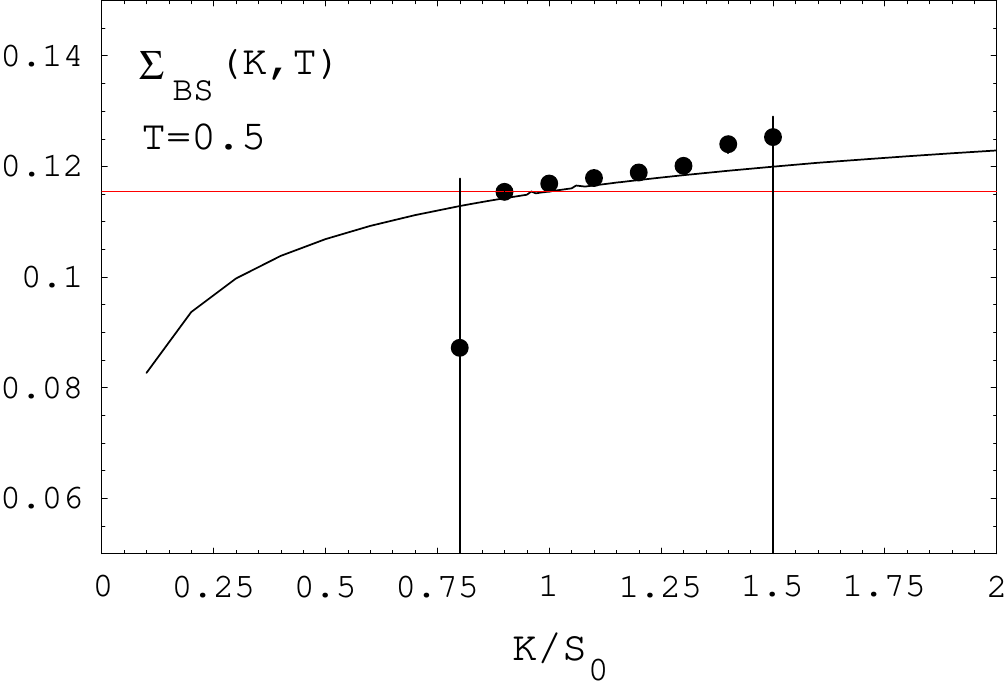}
\includegraphics[height=25mm]{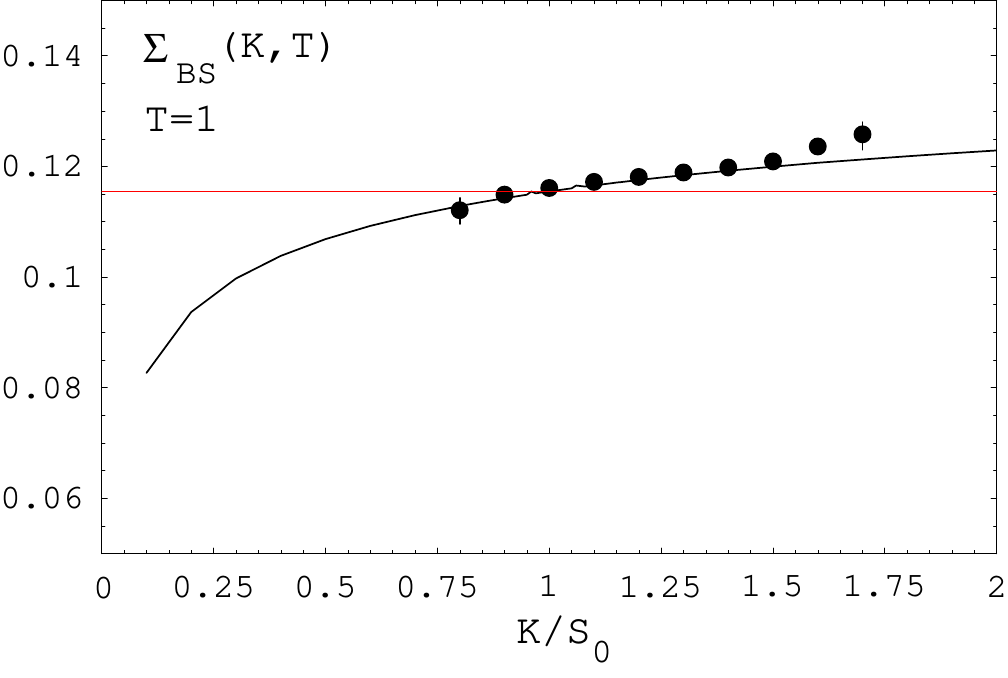}
\includegraphics[height=25mm]{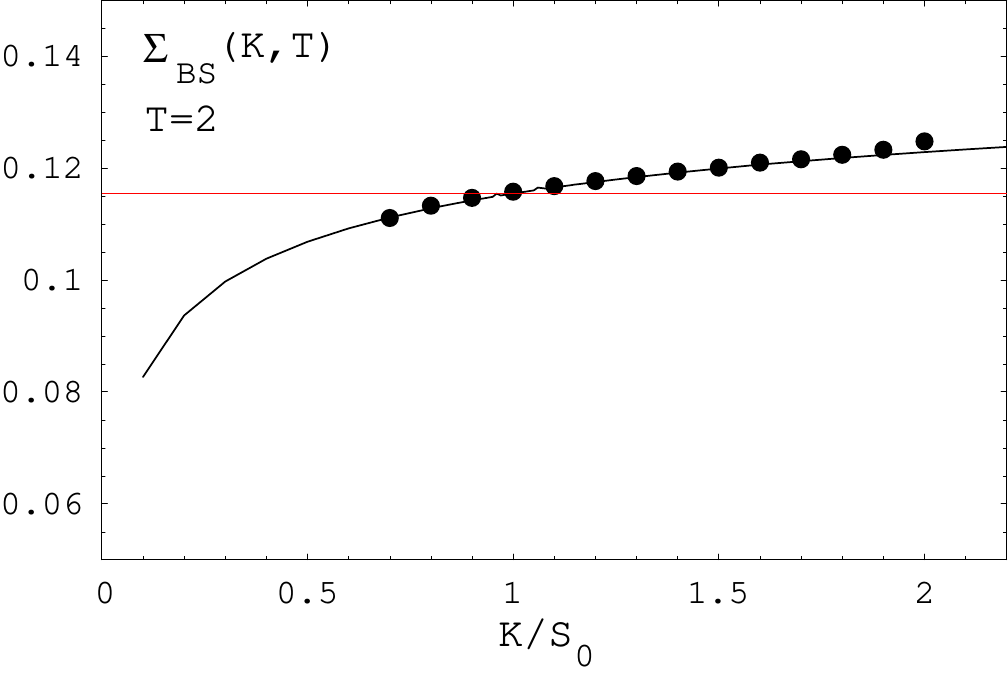}
\end{center}
\caption{
The equivalent log-normal volatility $\Sigma_{\rm LN}(K,S_0)$ 
of Asian options in the Black-Scholes model given by (\ref{SigBS}) 
(black curve).
The red line is at $\frac{1}{\sqrt3}\sigma$ and corresponds to the ATM  
equivalent volatility. The dots show the log-normal equivalent
volatility obtained by Monte Carlo pricing of the Asian options with maturity 
$T=0.5,1,2$. The BS model parameters are $r=q=0,\sigma=0.2$. The time step 
of the MC simulation is $\tau=0.01$ and the number of paths $N_{\rm MC}=1m$.}
\label{Fig:sigLN}
\end{figure}


Finally, the integral appearing in $\Lambda[f]$ can be computed in closed
form for each solution, and we have
\begin{align}
& T_1(\delta,\rho) = \int_0^1 dx e^{f_1(x)} = \frac{1}{\delta} \sinh\delta +
\frac{2\rho}{\delta^2} \sinh^2\left(\frac12\delta\right), 
\\
& T_2(\xi,\rho) = \int_0^1 dx e^{f_2(x)} = \frac{1}{2\xi}
\sin(2\xi) \left( 1 + \frac{\rho}{2\xi} \tan\xi \right) \,.
\end{align}

Substituting these results into (\ref{Lambda2}), we find the following
results for the function $\lambda(a,b;\rho)$ 
\begin{align}
& \lambda_1(a,b;\rho) 
\\
&= 
-2 a T_1(\delta) + a e^\delta \left(\frac{1+\gamma}{e^\delta+\gamma}\right)^2
 + \frac{1}{b^2} \rho \left(\delta - 2 \log\left(\frac{e^\delta + \gamma}{1+\gamma}\right)
 \right) - \frac{\rho^2}{b^2}
\nonumber
\\
& = a \left\{ 1 + \sinh^2\frac{\delta}{2} \left( 1 - \frac{4\rho}{\delta^2}
+ \frac{\rho^2}{\delta^2} \right) - \frac{2-\rho}{\delta} \sinh \delta \right\}
+ \frac{2}{b^2}\rho \log 
\left[\cosh\frac{\delta}{2} + \frac{\rho}{\delta} \sinh\frac{\delta}{2} \right]
- \frac{\rho^2}{b^2}, \nonumber\\
& \lambda_2(a,b;\rho) 
\\
&= 
-2 a T_2(\xi) + a \frac{\cos^2\eta}{\cos^2(\xi+\eta)}
 + \frac{1}{b^2} \rho \log\frac{\cos^2\eta}{\cos^2(\xi+\eta)} 
- \frac{\rho^2}{b^2}
\nonumber
\\
& = a \left\{ 1 - \sin^2 \xi \left(1 + \frac{\rho}{\xi^2} - 
\frac{\rho^2}{4\xi^2}\right) + \frac{\rho-2}{2\xi} \sin(2\xi)
\right\} 
 + \frac{2\rho}{b^2} \log \left[\cos \xi + \frac{\rho}{2\xi}\sin\xi\right]
- \frac{\rho^2}{b^2}\,, \nonumber
\end{align}
where $\delta$ and $\xi$ are given by the solutions of the equations
(\ref{deltaeq}) and (\ref{lambdaeq}), respectively.
For given $(a>0,b,\rho)$, only one of these two equations has a solution,
which determines the optimal function $f(x)$ uniquely, and the function
$\lambda(a,b;\rho)$. 
This completes the proof of Proposition~\ref{VarProblem}.
\end{proof}


\begin{proof}[Proof of Proposition \ref{prop:rhogeneral}]

The variational problem (\ref{RateFunction}) can be written equivalently 
in terms of $\mathcal{J}(x,\rho)$ defined as in (\ref{Jcaldef}), by 
introducing the function $f(y) = \sqrt{2\beta} g(y)$ as
\begin{equation}
\mathcal{J}(x,\rho)=
\inf_{f\in\mathcal{AC}_{0}[0,1],\int_{0}^{1} e^{f(y)}dy=\frac{x}{S_0}}
\frac12 \int_{0}^{1} (f'(y)-\rho)^{2} dy\,.
\end{equation}

The integral constraint on $f(y)$ is taken into account
by introducing a Lagrange multiplier $a$ and defining an auxiliary
functional
\begin{align}
\Lambda[f] &= \frac12 \int_0^1 dy 
\left( f'(y) - \rho \right)^2 + a\left(
\int_0^1 dy e^{f(y)} - \frac{x}{S_0} \right) \\
&= \frac12 \int_0^1 dy [f'(y)]^2 + 
a \int_0^1 dy e^{f(y)} - 
\rho f(1) + \frac12 \rho^2 - a \frac{x}{S_0}\,.\nonumber
\end{align}

The solution of this variational problem $f(y)$ satisfies the 
Euler-Lagrange equation
\begin{equation}\label{EL2}
f''(y) = a e^{f(y)},
\end{equation}
with boundary conditions (the condition at $y=1$ is a transversality condition)
\begin{equation}
f(0)=0\,,\qquad f'(1) = \rho\,.
\end{equation}
This differential equation and the associated boundary conditions are identical
to the equation appearing in the proof of Proposition~\ref{VarProblem}. As shown,
this can be solved exactly, and the solutions are given in (\ref{f1}), (\ref{f2}).
The details of the proof will be slightly different, as in the present 
case the coefficient
$a$ (the Lagrange multiplier) is not known, but is one of the unknowns of the
variational problem. However, we will show that it can be determined using the
integral constraint 
\begin{equation}\label{fconst}
\int_0^1 e^{f(y)} dy = \frac{K}{S_0}\,.
\end{equation}
Before proceeding with the solution of the variational problem, we give a
preliminary result which expresses the rate function only in terms of $a,f(1)$.

\begin{lemma}
The rate function $\mathcal{J}(x,\rho)$ is given by
\begin{eqnarray}\label{Jsimple}
\mathcal{J}(K/S_0,\rho) = 
a \left(\frac{K}{S_0} - e^{f(1)}\right) - \rho f(1) + \rho^2\,.
\end{eqnarray}
\end{lemma} 

\begin{proof}
The Euler-Lagrange equation (\ref{EL2}) conserves the following quantity 
\begin{equation}
E = \frac12 (f'(y))^2 - a e^{f(y)}\,,
\end{equation}
which gives
\begin{equation}
\frac12 [f'(y)]^2 - a e^{f(y)} = \frac12\rho^2 - a e^{f(1)} \,.
\end{equation}
Taking the integral of this relation over $x:(0,1)$, and using the 
constraint (\ref{fconst}) gives the result (\ref{Jsimple}).

\end{proof}

The only remaining part of the proof is determining $a,f(1)$. This can be
done from the constraint (\ref{fconst}). Substituting (\ref{f1}) into this
constraint gives
\begin{equation}\label{betaeq3}
\int_0^1 dx e^{f_1(x)} = \frac{1}{\delta} \sinh \delta + 
\frac{2\rho}{\delta^2} \sinh^2(\delta/2) = \frac{K}{S_0}\,.
\end{equation}
which is an equation for $\delta$. 
This equation has solutions only for $K/S_0 \geq 1+ \rho/2$. Once $\delta,\gamma$
are known, the Lagrange multiplier $a$ is determined using the relation
(\ref{f1const}). Substituting into (\ref{Jsimple}) we find the rate function
\begin{align}\label{J1res}
\mathcal{J}(K/S_0,\rho) &= \frac12 (\beta^2 - \rho^2) \left(1 - \frac{2\tanh(\beta/2)}
{\beta + \rho \tanh(\beta/2)} \right) \\
&\qquad\qquad  
- 2\rho \log \left[\cosh(\beta/2) + 
\frac{\rho}{\beta} \sinh(\beta/2) \right] + \rho^2 \,.\nonumber
\end{align}

A similar calculation using $f_2(x)$ gives
\begin{equation}\label{lambdaeq3}
\int_0^1 dx e^{f_2(x)} = 
\frac{1}{2\xi} \sin(2\xi) 
\left( 1 + \frac{\rho}{2} \frac{\tan\xi}{\xi} \right)
= \frac{K}{S_0}\,.
\end{equation}
Both $\eta$ and $\xi+\eta$ must be in the $(-\pi/2, \pi/2)$  range.
The equation (\ref{lambdaeq3}) has solutions only for $K/S_0 \leq 1 + \rho/2$.

Using the solution for $\xi$, the Lagrange multiplier $a$ is found from
(\ref{f2const}). Substituting into (\ref{Jsimple}) we find the rate function 
\begin{equation}\label{J2res}
\mathcal{J}(K/S_0,\rho) = 2\left(\xi^2 + \frac14 \rho^2\right) 
\left\{
\frac{\tan\xi}{\xi + \frac{\rho}{2}\tan\xi}  - 1
 \right\} - 
2\rho \log\left( \cos\xi + \frac{\rho}{2\xi} \sin\xi \right) 
+ \rho^2 \,.
\end{equation}
This completes the proof of Proposition~\ref{prop:rhogeneral}.

\end{proof}


\begin{proof}[Proof of Proposition \ref{FloatProp}]
Start by noting that $S_{t_{n}}=S_{0}e^{\sigma Z_{t_{n}}+(r-q-\frac{1}{2}\sigma^{2})t_{n}}$
can be written equivalently as $S_{0}e^{\frac{\sqrt{2\beta}}{n}\sum_{j=1}^{n}(V_{j}+\frac{\rho}{\sqrt{2\beta}})-\frac{\beta}{n}}$
in distribution where $V_{j}$ are i.i.d. $N(0,1)$ random variables. Let us also recall that
$A_{n}=\frac{1}{n}S_{0}\sum_{k=0}^{n-1}e^{\frac{\sqrt{2\beta}}{n}\sum_{j=1}^{k}(V_{j}+\frac{\rho}{\sqrt{2\beta}})
-\frac{\beta k}{n^{2}}}$. The terms $\frac{\beta}{n}$, $\frac{\beta k}{n^{2}}$ are uniformly bounded 
and negligible
and if we let $g(x)=\frac{1}{n}\sum_{j=1}^{\lfloor xn\rfloor}(V_{j}+\frac{\rho}{\sqrt{2\beta}})$,
then 
$\kappa e^{\frac{\sqrt{2\beta}}{n}\sum_{j=1}^{n}(V_{j}+\frac{\rho}{\sqrt{2\beta}})}
-\frac{1}{n}\sum_{k=0}^{n-1}e^{\frac{\sqrt{2\beta}}{n}\sum_{j=1}^{k}(V_{j}+\frac{\rho}{\sqrt{2\beta}})}
= \kappa e^{\sqrt{2\beta}g(1)} - \int_{0}^{1}e^{\sqrt{2\beta}g(x)}dx$.
The map $g\mapsto \kappa e^{\sqrt{2\beta}g(1)} - \int_{0}^{1}e^{\sqrt{2\beta}g(x)}dx$ is continuous
in the supremum norm and by contraction principle, $\mathbb{P}(\kappa S_{t_{n}}- A_{n}\in\cdot)$
satisfies a large deviation principle with the rate function
\begin{equation}\label{FloatRate}
\mathcal{H}(x)=\inf_{g\in\mathcal{AC}_{0}[0,1],
\kappa e^{\sqrt{2\beta}g(1)}-\int_{0}^{1}e^{\sqrt{2\beta}g(y)}dy=\frac{x}{S_{0}}}
\frac{1}{2}\int_{0}^{1}\left(g'(y)-\frac{\rho}{\sqrt{2\beta}}\right)^{2}dy.
\end{equation}
As $n\rightarrow\infty$, $\kappa S_{t_{n}}- A_{n}\rightarrow \kappa S_{0}e^{\rho}-S_{0}\frac{e^{\rho}-1}{\rho}$ a.s.
When $\kappa < \frac{1}{\rho}(1-e^{-\rho})$, the call option is out-of-the-money and 
\begin{equation}
C(n)=e^{-n\mathcal{H}(0)+o(n)},\qquad\text{as $n\rightarrow\infty$},
\end{equation}
and by put-call parity, when $\rho\neq 0$,
\begin{align}
C(n)-P(n)&=e^{-rt_{n}}\mathbb{E}\left[\kappa S_{t_{n}}- A_{n}\right]
\\
&=\kappa S_{0}e^{-\frac{r}{r-q}\rho}- e^{-\frac{r}{r-q}\rho}S_{0}\frac{e^{\rho}-1}{n(1-e^{-\frac{\rho}{n}})}
\nonumber
\\
&=\kappa S_{0}e^{-\frac{r}{r-q}\rho}- S_{0}e^{-\frac{r}{r-q}\rho}\frac{e^{\rho}-1}{\rho}
-\frac{e^{-\frac{r\rho}{r-q}}S_{0}(e^{\rho}-1)}{2n}+O(n^{-2}).
\nonumber
\end{align}
Therefore, as $n\rightarrow\infty$, the asymptotics for in-the-money put option is
\begin{equation}
P(n)=-\kappa S_{0}e^{-\frac{r}{r-q}\rho}+ S_{0}e^{-\frac{r}{r-q}\rho}\frac{e^{\rho}-1}{\rho}
+\frac{ e^{-\frac{r\rho}{r-q}}S_{0}(e^{\rho}-1)}{2n}+O(n^{-2}).
\end{equation}
When $\rho=0$, 
\begin{equation}
P(n)=(1-\kappa)S_{0}+e^{-n\mathcal{H}(0)+o(n)},
\qquad
\text{as $n\rightarrow\infty$}.
\end{equation}

When $\kappa=\frac{1}{\rho}(1-e^{-\rho})$, i.e., at-the-money, the asymptotics for $C(n)$ and $P(n)$ 
are governed by the central limit theorem. $\frac{1}{\sqrt{n}S_{0}}(\kappa S_{t_{n}}-A_{n})$ can be 
approximated by
\begin{equation}
\kappa e^{\rho}\frac{\sqrt{2\beta}}{\sqrt{n}}\sum_{j=0}^{n-1}V_{j}-
\frac{\sqrt{2\beta}}{n^{3/2}} \sum_{j=0}^{n-1}V_{j}\sum_{i=j+1}^{n}e^{\rho\frac{i}{n}},
\end{equation}
with $V_j = N(0,1)$ i.i.d. random variables. The variance of this expression converges to 
\begin{align}
\int_{0}^{1}\left[\kappa e^{\rho}\sqrt{2\beta}-\frac{\sqrt{2\beta}}{\rho}(e^{\rho}-e^{\rho x})\right]^{2}dx
&=\frac{2\beta}{\rho^{2}}\int_{0}^{1}(1-e^{\rho x})^{2}dx
\\
&=\frac{2\beta}{\rho^{2}}\left[1-\frac{2}{\rho}(e^{\rho}-1)+\frac{e^{2\rho}-1}{2\rho}\right].
\nonumber
\end{align}

We can further use the nonuniform Berry-Esseen bound for the central limit 
theorem to obtain the following asymptotics,
\begin{equation}
\lim_{n\rightarrow\infty}\sqrt{n}C(n)
=\lim_{n\rightarrow\infty}\sqrt{n}P(n)
=S_{0}e^{-\frac{r\rho}{r-q}}\mathbb{E}[Z1_{Z\geq 0}],
\end{equation}
where $Z$ is a normal random variable with mean $0$ and variance
\begin{equation}
\frac{2\beta}{\rho^{2}}\left[1-\frac{2}{\rho}(e^{\rho}-1)+\frac{e^{2\rho}-1}{2\rho}\right].
\end{equation}

When $\kappa > \frac{1}{\rho}(1-e^{-\rho})$, the put option is out-of-the-money and
\begin{equation}
P(n)=e^{-n\mathcal{H}(0)+o(n)},\qquad\text{as $n\rightarrow\infty$},
\end{equation}
and when $\rho\neq 0$, we have for in-the-money call option
\begin{equation}
C(n)=\kappa S_{0}e^{-\frac{r}{r-q}\rho}-S_{0}e^{-\frac{r}{r-q}\rho}\frac{e^{\rho}-1}{\rho}
-\frac{e^{-\frac{r\rho}{r-q}}S_{0}(e^{\rho}-1)}{2n}+O(n^{-2}),
\end{equation}
and when $\rho=0$,
\begin{equation}
C(n)=S_{0}(\kappa - 1)+e^{-n\mathcal{H}(0)+o(n)},
\qquad
\text{as $n\rightarrow\infty$}.
\end{equation}
\end{proof}



\begin{proof}[Proof of Proposition~\ref{prop:impvol}]

i) The price of an undiscounted European option in the Black-Scholes model 
depends only on $\sigma^2 T$ and $K/F$ with $F$ the forward asset price.
In our case given by (\ref{CAsian}) we have $F = A_\infty$, and we denote 
this dependence as
$e^{-rT} A_\infty \bar C_{\rm BS}(K/A_\infty,\sigma^2 T)$,
with $\bar C_{\rm BS}(k, v) := \Phi(\frac{1}{\sqrt{v}}(-\log k + \frac12v)) -
k \Phi(\frac{1}{\sqrt{v}}(-\log k - \frac12v))$.

By definition of the 
equivalent log-normal implied volatility we have
\begin{eqnarray}
C(n) = e^{-rT} A_\infty \bar C_{\rm BS}(K/A_\infty,\Sigma_{\rm LN}^2 T) \,.
\end{eqnarray}

Consider an OTM Asian call option $K> A_\infty$.
We have from Proposition~\ref{prop:13}
\begin{eqnarray}
\lim_{n\to \infty} \frac{1}{n} \log C(n) = 
- \frac{1}{2\beta} \mathcal{J}(K/S_0,\rho) \,.
\end{eqnarray}
Also, we have
\begin{eqnarray}
\lim_{T\to 0} (\Sigma^2_{\rm LN} T) \log \left( A_\infty
\bar C_{\rm BS}(K/A_\infty,\Sigma_{\rm LN}^2 T)\right) = 
- \frac12 \log^2(K/A_\infty)
\end{eqnarray}


We get thus, setting $T=t_n$,
\begin{eqnarray}
\lim_{n\to \infty} \Sigma^2_{\rm LN}(K,n) n^2 \tau &=& 
\lim_{n\to \infty} 
\frac{\Sigma^2_{\rm LN}(K,n) n\tau 
\log[A_\infty\bar C_{\rm BS}(K/A_\infty,\Sigma^2_{\rm LN}T)]}
{\frac{1}{n} \log C(n)} \\
&=&
\beta \frac{\log^2(K/A_\infty)}{\mathcal{J}(K/S_0,\rho)}\,.\nonumber
\end{eqnarray}
Recalling that $\beta = \frac12 \sigma^2 n^2\tau$ this is written
equivalently as
\begin{equation}
\lim_{n\to \infty} \frac{1}{\sigma^2} \Sigma_{\rm LN}^2(K,n) = 
\frac12 \frac{\log^2(K/A_\infty)}
{\mathcal{J}(K/S_0, \rho)}\,,
\end{equation}
which reproduces the result (\ref{SigBS}).

ii) At-the-money Asian option.
The Black-Scholes formula  gives for this case
\begin{align}
\bar C_{\rm BS}(1,\Sigma_{\rm LN}^2 T) &=
\Phi\left(\frac12 \Sigma_{\rm LN}\sqrt{T}\right)
- \Phi\left(-\frac12 \Sigma_{\rm LN}\sqrt{T}\right)  \\
&=
\frac{1}{\sqrt{2\pi}} \Sigma_{\rm LN} \sqrt{T} \left(1 + 
O\left(\Sigma_{\rm LN}^2 T\right)\right)\,.\nonumber
\end{align}
The large $n$ asymptotics of the ATM Asian option given in 
Proposition~\ref{prop:ATM} reads
\begin{equation}
C(A_\infty,n) = \frac{1}{\sqrt{\pi}} S_0 e^{-\frac{r\rho}{r-q}}
\sqrt{\beta v(\rho)} \frac{1}{\sqrt{n}}\,.
\end{equation}

The two results are related as $C(A_\infty,n) = e^{-rT} A_\infty
\bar C_{\rm BS}(1,\Sigma^2_{\rm LN} T)$.
Recalling that we have $\sigma\sqrt{t_n}=
\frac{1}{\sqrt{n}} \sqrt{2\beta}$ we obtain the asymptotics of the equivalent 
implied volatility of the ATM Asian option
\begin{equation}
\lim_{n\to \infty} \frac{\Sigma_{\rm LN}(A_\infty,n)}{\sigma} = 
\frac{S_0}{A_\infty} \sqrt{v(\rho)}\,.
\end{equation}
This reproduces equation (\ref{ATM}). 


The proof of (\ref{SigN}) proceeds in a similar way, starting with the 
Bachelier formula for the call option prices.

\end{proof}

\end{document}